\documentclass{article}

\usepackage{arxiv}
\usepackage[T1]{fontenc}
\usepackage[utf8]{inputenc}
\usepackage{textcomp}
\usepackage{amsmath}
\usepackage{amssymb}
\usepackage{graphicx}
\usepackage{booktabs}
\usepackage{xcolor}
\usepackage{algorithm}
\usepackage{algpseudocode}
\usepackage{float}
\usepackage{placeins}
\usepackage{tikz}
\usetikzlibrary{arrows.meta,positioning}
\usepackage{caption}
\usepackage{csquotes}
\usepackage[backend=biber,style=numeric,sorting=none]{biblatex}
\usepackage[hidelinks]{hyperref}
\usepackage{orcidlink}

\graphicspath{{figures/}}

\newtheorem{theorem}{Theorem}[section]
\newcommand{\aiuse}[1]{\par\medskip\noindent\textbf{Declaration of use of AI.}\enskip #1}
\newcommand{\ethics}[1]{\par\medskip\noindent\textbf{Ethics.}\enskip #1}
\newcommand{\dataccess}[1]{\par\medskip\noindent\textbf{Data accessibility.}\enskip #1}
\newcommand{\aucontribute}[1]{\par\medskip\noindent\textbf{Authors' contributions.}\enskip #1}
\newcommand{\competing}[1]{\par\medskip\noindent\textbf{Competing interests.}\enskip #1}
\newcommand{\funding}[1]{\par\medskip\noindent\textbf{Funding.}\enskip #1}
\newcommand{\ack}[1]{\par\medskip\noindent\textbf{Acknowledgements.}\enskip #1}

\title{Quantifying the complexity of trajectory ensembles with clustering-weighted multivariate multiscale sample entropy}

\author{%
  Chenxiao Tian~\orcidlink{0000-0001-6468-8201} \\
  Princeton University \\
  Princeton, NJ 08544
  \And
  J\"urgen Hackl~\orcidlink{0000-0002-8849-5751}\thanks{Corresponding author.} \\
  Princeton University \\
  Princeton, NJ 08544 \\
  \texttt{hackl@princeton.edu}
}
\date{\today}

\renewcommand{\shorttitle}{Complexity of trajectory ensembles (CWMMSE)}
\renewcommand{\headeright}{}
\renewcommand{\undertitle}{}

\hypersetup{
  pdftitle={Quantifying the complexity of trajectory ensembles with clustering-weighted multivariate multiscale sample entropy},
  pdfauthor={Chenxiao Tian, Juergen Hackl},
  pdfkeywords={sample entropy, multiscale entropy, trajectory clustering, weighted entropy, ensemble complexity, functional diversity},
}

\begin{document}
\maketitle

\begin{abstract}
Across the physical and life sciences, data increasingly appear as ensembles of trajectories, from chaotic flows and satellite constellations to clinical cohorts. Established sample-entropy measures characterize individual time series, while averaging across an ensemble discards population structure and cannot distinguish redundancy from diversity. We introduce clustering-weighted multivariate multiscale sample entropy (CWMMSE), which groups trajectories into behavioral patterns and weights each by its dynamical complexity. CWMMSE is a weighted entropy of the population's pattern distribution. Its empirical plug-in estimator is strongly consistent for a fixed finite partition, and it separates two components that can diverge in real data: individual complexity and population diversity. Both are essential. Averaging ignores diversity, whereas spread alone can mistake a varied but predictable population for a complex one. Across eleven physical, environmental, engineering, and biomedical systems, CWMMSE ranks a calm ocean region above an energetic but individually more complex one, identifies a major earthquake as a collapse in system complexity, and reverses the conclusion from averaging in cardiac cohorts, where disease reduces population diversity. Supported by an open, reproducible implementation, these results show that population complexity should be measured rather than averaged.
\end{abstract}


\begin{refsection}

\section{Introduction}\label{sec:intro}

Across the physical, biological, medical and engineering sciences, a system is increasingly observed not
through a single record but through a large ensemble of the many trajectories it produces. Yet we still
lack a principled way to measure the complexity of that ensemble as a whole. Quantifying the complexity of a
system known only through its observed dynamics is itself a long-standing task: when the governing
equations are unavailable or too high-dimensional to be useful, the measured record alone must carry the
information. Entropy-rate estimators have become the established approach. They quantify the
rate at which a signal generates new information, equivalently the unpredictability of its next value
given its recent past. Approximate entropy~\cite{pincus1991} introduced a practical finite-sample
estimator suitable for short, noisy records; sample entropy~\cite{richman2000} removed its self-matching
bias; multiscale entropy~\cite{costa2002,costa2005} recognised that complexity lives across temporal
scales, distinguishing the rich structure of healthy physiology from the scale-collapsed regularity of
disease or the flatness of noise; and multivariate multiscale entropy~\cite{ahmed2011} extended the
construction to coupled channels. These methods sit within a broad family~\cite{wu2013,humeau2015}
alongside permutation entropy~\cite{bandt2002,riedl2013}, distribution entropy for short
records~\cite{li2015disten}, Lempel-Ziv complexity~\cite{lempel1976} and
the complexity-entropy plane that separates chaos from noise~\cite{rosso2007}, all resting on the
reconstruction of a state space by delay embedding (representing the system's hidden state by short stretches of the observed record)~\cite{takens1981,grassberger1983,kantz2004}. They have
been applied from neural recordings~\cite{stam2005} to the now-classical loss of physiological complexity
with disease and ageing~\cite{goldberger2002}.

Every one of these measures characterises a \emph{single} time series. Ensembles, by contrast, are
everywhere: the routes of a vehicle fleet, the orbits of a satellite constellation together with the
debris of a fragmentation event, the electrocardiograms of a clinical cohort, the Lagrangian paths of
particles in a turbulent flow, or the light curves of a variable-star survey. For such data the
object of interest is not the complexity of any single trajectory but the complexity of the system that
produced them all.

The obvious response, to compute a per-trajectory measure and average it over the ensemble, discards
exactly the information that makes the population a system. Consider two ensembles with identical mean
per-trajectory entropy: in the first, one intricate behaviour is repeated by every member; in the second,
many distinct but individually simple behaviours coexist. The first is redundant and the second diverse,
yet the mean cannot tell them apart. A faithful measure of system complexity must therefore combine two
ingredients that the average conflates: the \emph{within-trajectory} complexity of the members and the
\emph{between-trajectory} diversity of the population. The second is naturally captured by clustering the
trajectories into recurring behavioural patterns, a task for which a mature time-series clustering
literature provides the tools~\cite{aghabozorgi2015,paparrizos2015,lee2007,zheng2015}, and asking how the
population distributes over those patterns.

We are aware of no single summary that supplies both ingredients at once. Standard
practice in physiology and beyond computes a per-record measure and reports its cohort mean or
distribution~\cite{costa2002,goldberger2002}, characterising the typical member but, as argued
above, remaining blind to the ensemble's diversity. Recent ensemble variants of permutation and multiscale
entropy pool records to improve the \emph{estimation} of a single series rather than to measure the
complexity of a population~\cite{wu2013,humeau2015}. Conversely, the diversity of a population is studied
in its own right in ecology, where Rao's quadratic entropy and similarity-sensitive indices weight a
population's spread by pairwise trait dissimilarities, that is, by how different its members
are from one another; this line of work includes the functional-diversity and Hill-number
frameworks~\cite{rao1982,bottadukat2005,ricottaszeidl2006,leinstercobbold2012,
pavoinebonsall2011,petchey2002,villeger2008,chao2014}; but a measure of spread alone cannot distinguish a varied population of intricate
behaviours from an equally varied population of trivial ones. What is missing is a single quantity that
is simultaneously sensitive to how irregular the members are and to how varied the population is, and
that reduces to each of these established notions in the appropriate limit.

Here we introduce \emph{clustering-weighted multivariate multiscale sample entropy} (CWMMSE), which fuses
the two ingredients into a single number. Each trajectory is assigned its multivariate multiscale sample
entropy; the trajectories are grouped into patterns by clustering; and each pattern's mean complexity is
weighted by its contribution to the population's Shannon diversity. The construction is not ad hoc: the
result is exactly the \emph{weighted entropy} of Belis and Guia\c{s}u~\cite{belis1968,guiasu1971}, an
entropy of the population's pattern distribution in which each outcome counts in proportion to a utility,
here its own dynamical complexity. CWMMSE therefore has a clear information-theoretic meaning yet is no
harder to compute than the entropy and clustering steps it combines. It is also distinct from
the established diversity indices it superficially resembles, such as Rao's quadratic
entropy~\cite{rao1982,ricottaszeidl2006} and similarity-sensitive diversity~\cite{leinstercobbold2012},
which weight by pairwise dissimilarities rather than by intrinsic complexity. This distinction proves
decisive: a maximally spread but individually predictable ensemble is maximally diverse to those indices,
yet CWMMSE correctly returns near zero (\S\ref{sec:results}).

This framing makes a concrete, testable prediction: system complexity decomposes into individual
complexity and ensemble diversity, and the two genuinely diverge. We find ensembles that are individually
complex yet collectively redundant, ensembles that are maximally diverse yet individually trivial, and,
most notably, diseased clinical cohorts that match healthy ones in per-record complexity while
collapsing in ensemble diversity. In that last case the conventional averaged entropy not only
fails to detect the difference but orders the cohorts incorrectly, whereas CWMMSE reverses this.
Because both ingredients matter, measures built on either one alone fail: averaged entropy is blind to
diversity, while diversity indices, including the cluster entropy and Rao's quadratic entropy, rank a
cloud of perfectly predictable satellite orbits as the most complex system of all.

Our contributions are fourfold. First, we define CWMMSE and identify it as the Belis and Guia\c{s}u
weighted entropy of an ensemble's pattern distribution, together with a strong-consistency result for its empirical plug-in estimator under a fixed finite pattern partition
(\S\ref{sec:methods}). Second, we show that it decomposes system complexity into within- and
between-trajectory components, and that these components diverge both in a controlled construction and in
field data. Third, we establish that both ingredients are necessary by means of two controls that
single-ingredient measures fail, one of which, a population of highly diverse but individually predictable
trajectories, also separates CWMMSE precisely from functional-diversity indices such as Rao's quadratic entropy.
Fourth, we demonstrate the measure across eleven systems spanning the physical, environmental, engineering and
biomedical sciences, with subsample stability ranges and surrogate-data nulls.

The paper is organised as follows. The methods (\S\ref{sec:methods}) develop the measure and its
underpinning theory: the within-trajectory complexity, the clustering that exposes the population's
recurring patterns, their fusion into a weighted entropy, and the theorem giving its large-sample
behaviour. The results (\S\ref{sec:results}) evaluate the measure, from validation on deterministic chaos,
through demonstrations on field data and two decisive controls, to a cross-domain survey of eleven systems
and a robustness analysis. The discussion (\S\ref{sec:discussion}) relates CWMMSE to entropy-rate and
functional-diversity measures and delimits its scope and the questions it leaves open, and the conclusions
(\S\ref{sec:conclusions}) synthesise the argument. The full algorithm, the consistency proof, the complete
benchmark and robustness studies, and a dataset-by-dataset account are collected in the electronic
supplementary material.

\section{Material and methods}\label{sec:methods}

\begin{figure}[tb]
\centering
\includegraphics[width=\linewidth]{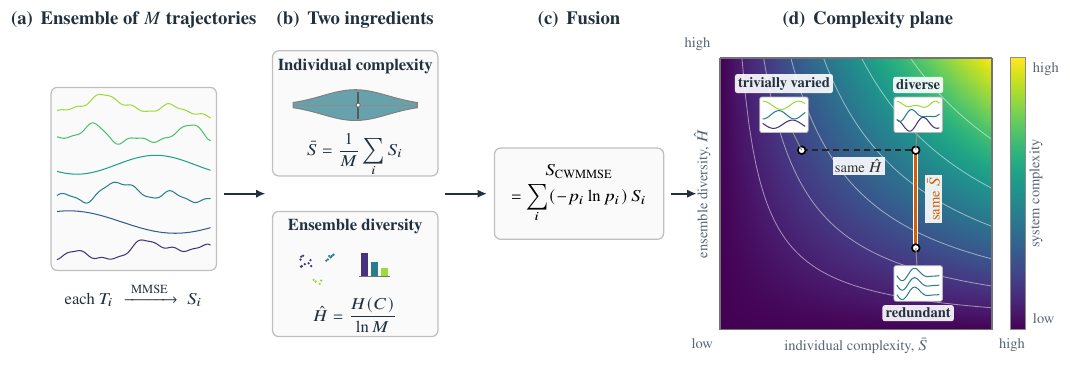}
\caption{CWMMSE is built in four steps, fusing within-trajectory complexity and ensemble diversity into a single weighted entropy. (a) A system is observed as an ensemble of $M$ multivariate trajectories $T_1,\dots,T_M$. (b) Two ingredients are evaluated together: the individual complexity $\bar S$ (mean within-trajectory MMSE), where each member's multivariate multiscale sample entropy (MMSE) is $S(T_j)$ and the mean is $\bar S=\frac{1}{M}\sum_j S(T_j)$; and the ensemble diversity $\hat H=H(C)/\ln M$ (normalised cluster entropy), found by clustering the trajectories into behavioural patterns with population shares $p_i$. (c) The two are fused into the weighted entropy $S_{\mathrm{CWMMSE}}=\sum_i(-p_i\ln p_i)\,S_i$ (equation~\eqref{eq:cwmmse}), the weighted entropy of Belis and Guia\c{s}u. (d) The complexity plane spanned by the two ingredients, used throughout the paper. The shading (dark for low, bright for high) and white iso-lines show the separable baseline $H(C)\,\bar S$ (in the later complexity-plane figures the marker colour instead encodes CWMMSE itself), to which CWMMSE adds the complexity--diversity coupling of equation~\eqref{eq:coupling}, so system complexity is high only when both ingredients are high. Three specimen ensembles make the point: the diverse ensemble (individually complex \emph{and} behaviourally varied) is high, whereas the redundant ensemble (one behaviour repeated, sharing the diverse ensemble's $\bar S$ along the vertical ``same~$\bar S$'' line) and the trivially varied ensemble (many but individually simple members, sharing its $\hat H$ along the horizontal ``same~$\hat H$'' line) are both low, so neither ingredient alone suffices.}
\label{fig:schematic}
\end{figure}

We observe a system as an ensemble $\mathcal{D}=\{T_1,\dots,T_M\}$ of $M$ multivariate time series, each
generated by the same underlying process under different initial conditions, individuals or windows of
observation. Every trajectory $T_j$ has $c$ channels and length $N_j$; lengths may differ across datasets, though within a given ensemble the trajectories share a common per-dataset window length (ESM), and
$M$ may be large. The $c$ channels are the problem-specific observables, and the per-dataset choices are
listed in the electronic supplementary material (ESM). A trajectory may be a complete object, such as the
velocity record of one drifter, or a fixed-duration window cut from a longer record; the two readings are
formally identical, differing only in whether the measured diversity is across coexisting members or
across behaviours one system visits over time. The construction proceeds in three stages
(Figure~\ref{fig:schematic}).

\subsection*{Within-trajectory complexity}
The first ingredient is the multivariate multiscale sample entropy (MMSE)\footnote{Throughout, MMSE
abbreviates \emph{multivariate multiscale sample entropy}; it is unrelated to the identically abbreviated
minimum mean-square error.} of each trajectory~\cite{ahmed2011,
costa2002}, built on the sample entropy of Richman and Moorman~\cite{richman2000}. Write $x^{(a)}_{l}$ for
sample $l$ of channel $a$ of the trajectory, for $a=1,\dots,c$. At temporal scale $s$ each channel is
coarse-grained by averaging successive non-overlapping blocks of $s$ samples,
\begin{equation}\label{eq:cg}
y^{(a),s}_{t}=\frac{1}{s}\sum_{l=(t-1)s+1}^{ts} x^{(a)}_{l},\qquad t=1,\dots,\lfloor N_j/s\rfloor ,
\end{equation}
where $y^{(a),s}_{t}$ is the $t$-th coarse-grained value of channel $a$ and $t$ indexes the coarse-grained
time points. On the coarse-grained channels we form composite delay vectors by delay embedding~\cite{takens1981},
stacking time-delayed copies of each channel so that a short stretch of the record stands in for the
system's unobserved state, with per-channel embedding dimensions $\mathbf{m}=(m_1,\dots,m_c)$ and delays
$\boldsymbol{\tau}=(\tau_1,\dots,\tau_c)$, of total dimension $d=\sum_{a=1}^{c} m_a$. After per-channel normalisation, two distinct
vectors are neighbours when their Chebyshev distance, the largest difference across their coordinates, is
at most a tolerance $r$. Let $B_d(r)$ be the fraction of neighbouring pairs at total embedding
dimension $d$, and $B_{d+c}(r)$ the fraction when every channel's embedding is extended by one further
sample (dimension $d+c$). Both fractions are counted over the template indices admissible at the larger dimension $d+c$,
so that their ratio lies in $[0,1]$. The sample entropy at scale $s$ is then
\begin{equation}\label{eq:sampen}
S^{(s)}(T_j)=-\ln\!\big(B_{d+c}(r)/B_d(r)\big),
\end{equation}
the negative log of the conditional probability that two segments matching at embedding dimension $d$ still match
when extended by one further sample per channel. This simultaneous all-channel extension is the convention
of our released implementation. It coincides with the multivariate sample entropy of Ahmed and
Mandic~\cite{ahmed2011} for a single channel and, for $c>1$, is a distinct construction that extends all $c$
channels jointly rather than one at a time (its embedding grows by $c$ at each step). We adopt it for channel
symmetry and implementation simplicity, and apply it identically across every system (ESM).
The within-trajectory complexity $S(T_j)$ is taken at a single scale or summed over scales as suits the
data, the per-system scale choice being listed in the ESM (several short-window systems use a single
scale). Because $r$ acts on normalised channels, $S$ is dimensionless and comparable across heterogeneous
data. The full algorithm, including the order $N_j\log N_j$ neighbour counting by a $k$-d tree, is in the ESM.

\subsection*{From trajectories to patterns}
The second ingredient is the diversity of the ensemble. Given a dissimilarity between trajectories we
partition $\mathcal{D}$ into clusters $C_1,\dots,C_k$ by agglomerative Ward linkage~\cite{ward1963} cut at a chosen
height~\cite{aghabozorgi2015}: trajectories are merged bottom-up into ever larger groups, and cutting the
resulting tree at a height fixes how dissimilar two trajectories may be while still counting as the same
behaviour. The cut height is an explicit \emph{resolution}: it sets the scale at which
two trajectories count as the same behaviour, and is not an attempt to recover a true number of clusters.
We express it as a fraction of the maximum linkage distance so that a single relative threshold transfers
across datasets. The dissimilarity itself is a modelling choice with a scale. Our default sums the
point-wise Euclidean distance over time, which depends on amplitude and length; within each ensemble this point-wise sum is evaluated over members sharing the common per-dataset window length $N$ (ESM), so the comparison is defined index by index. We verify in
\S\ref{sec:results} that every contrast also holds under a scale- and length-invariant normalisation.
Clustering, rather than a parametric mixture or a kernel density, is a deliberate choice. It assumes
nothing about the geometry of the space of trajectories, accommodates the heterogeneous, unequal-length
and irregularly sampled records that real ensembles present, and exposes the population's pattern
distribution directly as the cluster proportions $p_i$. Its one cost is the resolution parameter,
which we therefore state explicitly and report across a range rather than fixing to a single value.
Ward linkage is our default, chosen for its within-cluster-variance criterion; because that criterion is
strictly appropriate to a squared-Euclidean representation whereas our default dissimilarity is a general
point-wise distance, we verify that average linkage, which is compatible with a general dissimilarity,
yields the same featured contrasts (\S\ref{sec:results}), retaining average and complete linkage as
sensitivity checks.

\subsection*{Clustering-weighted MMSE}
The two ingredients now combine into a single number. Let $p_i=|C_i|/M$ be the fraction of trajectories in cluster $i$, and $S_i$ the mean within-trajectory
complexity of its members. The clustering-weighted multivariate multiscale sample entropy is
\begin{equation}\label{eq:cwmmse}
S_{\mathrm{CWMMSE}}=\sum_{i=1}^{k}\big(-p_i\ln p_i\big)\,S_i ,
\end{equation}
with the convention $0\ln 0=0$. Each pattern contributes its complexity $S_i$ weighted by $-p_i\ln p_i$,
its share of the population's Shannon diversity. The two ingredients in isolation are the mean MMSE
$\bar S=\sum_i p_i S_i$, the ensemble average that discards diversity and is the natural baseline, and the
Shannon cluster entropy $H(C)=-\sum_i p_i\ln p_i\in[0,\ln k]$, which we report through the
\emph{normalised cluster entropy} $\hat H=H(C)/\ln M\in[0,1]$, dividing by the maximum $\ln M$ attained when
every trajectory forms its own pattern, and which discards individual complexity. Because this denominator
grows with the ensemble size while $H(C)$ is bounded above by $\ln k$, $\hat H$ is a finite-sample
descriptive normalisation: under a fixed finite population partition it tends to zero as $M\to\infty$, so we
read it within a study rather than as an absolute quantity comparable across studies of differing ensemble
size unless recomputed at a common $M$.

\subsection*{CWMMSE is a weighted entropy}
Equation~\eqref{eq:cwmmse} is an instance of a classical functional. The weighted entropy of Belis and
Guia\c{s}u~\cite{belis1968,guiasu1971} attaches a non-negative utility $w_i$ to each outcome of a
distribution,
\begin{equation}\label{eq:weighted}
H^{w}=-\sum_{i=1}^{k} w_i\,p_i\ln p_i ,
\end{equation}
and identifying $w_i$ with the within-cluster complexity $S_i$ turns equation~\eqref{eq:weighted} into
equation~\eqref{eq:cwmmse} exactly. CWMMSE is thus the weighted Shannon entropy~\cite{shannon1948,
cover2006} of the population's pattern distribution, each outcome counting in proportion to how
dynamically complex the trajectories realising it are. Weighted entropy was used by
Guia\c{s}u~\cite{guiasu1986} to group data into classes, exactly the clustered-ensemble setting here, and
is the across-ensemble counterpart of the within-series weighted permutation
entropy~\cite{fadlallah2013}, which weights each ordinal pattern of one series by its local amplitude
variance. The novelty of CWMMSE is the choice of dynamical complexity as the weight, the clustered-ensemble
application, and the following large-sample consistency result.

\begin{theorem}[Consistency of empirical CWMMSE at fixed resolution]\label{thm:weighted}
Fix a pattern resolution, so that a random trajectory takes values in a finite set of patterns
$i=1,\dots,k$ with fixed probabilities $p_i>0$, and, given the pattern, let the within-trajectory
complexities be independent and identically distributed with finite mean $S_i$. From $M$
independent draws form the empirical frequencies $\hat p_i$ and within-pattern means $\hat S_i$. Then, as
$M\to\infty$, $\sum_i(-\hat p_i\ln\hat p_i)\hat S_i\to-\sum_i p_i\ln p_i\,S_i$ almost surely. For
$S_i\equiv1$ this recovers convergence of the empirical Shannon entropy.
\end{theorem}

The proof (ESM) applies the strong law of large numbers to $\hat p_i$ and $\hat S_i$ and the continuous
mapping theorem. The result holds at a fixed resolution; it does not cover the regime in which the cut
height is held fixed as $M$ grows, where the number of patterns $k$ increases with $M$ and no fixed
$p_i>0$ exists (in the chaotic and several field ensembles $k$ is a large fraction of $M$). It likewise
treats the pattern assignment as a fixed measurable map with fixed probabilities $p_i$, and so does not
establish consistency of the sample-dependent Ward clustering actually used, whose partition is estimated
from the same data. Finite-sample
behaviour is characterised empirically below by subsampling, since each $\hat S_i$ is itself a
finite-length sample-entropy estimate rather than a converged entropy rate. The same argument gives identical consistency for the separable
product $H(C)\,\bar S$, so the result does not by itself distinguish the two.

\subsection*{Why this weighting}
Two degenerate limits show that both ingredients are indispensable and enter jointly. If a single pattern
dominates, $p_1\to1$, then $-p_1\ln p_1\to0$ and $S_{\mathrm{CWMMSE}}\to0$ however large $S_1$: a system
that endlessly repeats one behaviour, however intricate, is simple as a system. If every pattern is
trivial, $S_i\to0$, then $S_{\mathrm{CWMMSE}}\to0$ however diverse the population. The mean $\bar S$
violates the first limit and the cluster entropy $H(C)$ the second. A natural repair is the separable
product $H(C)\,\bar S$, which respects both limits but coincides with CWMMSE only when complexity and diversity are uncoupled
across clusters (in particular when complexity is homogeneous), since
\begin{equation}\label{eq:coupling}
\begin{aligned}
S_{\mathrm{CWMMSE}}-H(C)\,\bar S
&=\sum_{i=1}^{k}\big(-p_i\ln p_i\big)\big(S_i-\bar S\big)\\
&=\sum_{i=1}^{k}p_i\big(S_i-\bar S\big)\big(-\ln p_i-H(C)\big)\\
&=\operatorname{Cov}_p\!\big(S_i,\,-\ln p_i\big),
\end{aligned}
\end{equation}
the population-weighted covariance between each pattern's complexity $S_i$ and its surprisal $-\ln p_i$ (the
added $H(C)\sum_i p_i(S_i-\bar S)=0$ centres the surprisal, using $\mathbb{E}_p[-\ln p_i]=H(C)$ and
$\mathbb{E}_p[S_i]=\bar S$). CWMMSE
retains this coupling and the separable product discards it; we show in \S\ref{sec:results} that the
coupling is non-negligible in real data, which is the principled reason to prefer the single
weighted-entropy object over the product. The convex normalisation $S_{\mathrm{CWMMSE}}/H(C)$ is a
surprisal-weighted average of the cluster complexities (weights $-p_i\ln p_i$), distinct from the
population mean $\bar S$, and discards the diversity magnitude, so it is rejected. The weighted entropy also has a natural place among the
generalisations of Shannon's measure: where the R\'enyi~\cite{renyi1961} and Tsallis~\cite{tsallis1988}
families deform the entropy along an order parameter and recover Shannon as that order tends to one, the
weighted entropy deforms it along an orthogonal utility axis, here the dynamical complexity of each
pattern, and recovers Shannon when the utilities are equal. We do not claim that this form is the unique
functional satisfying the two limits. Its appeal
is that it is a single established, interpretable object that meets both limits and, unlike the separable
product, retains the complexity-diversity coupling. A fuller comparison of this functional family is given
in the ESM.

\subsection*{Implementation}
The entropy, clustering and weighting steps are collected in a single open-source Python package and
applied by the same procedure to each of the eleven systems we analyse (\S\ref{sec:results}). These differ
only in per-dataset settings, such as the embedding, tolerance, window length and clustering resolution
(all listed in the ESM), and not in the analysis itself, so the cross-domain comparison uses one measure
rather than eleven tuned variants. Every figure and number is regenerated from openly available data (Data
accessibility). The cost is dominated by the per-trajectory sample entropy, reduced to order $N_j\log N_j$ by
$k$-d-tree neighbour counting, and by the $O(M^2)$ trajectory dissimilarity; both are modest for the
ensembles studied here (up to $M\approx1700$), and the dissimilarity is readily parallelised.

\section{Results}\label{sec:results}

We evaluate CWMMSE on eleven systems drawn from the physical, environmental, engineering and biomedical
sciences (Table~\ref{tab:cross} lists each with its data source and contrast; per-system settings, sample
sizes and illustrations are in the ESM). They play four roles. Deterministic-chaos ensembles from the
Lorenz and R\"ossler systems~\cite{lorenz1963,rossler1976} provide a validation, since the complexity ordering across their regimes is
known in advance. Ocean drifters~\cite{elipot2016drifter} from two contrasting regions of the North Atlantic demonstrate the
decomposition, the divergence of individual complexity from ensemble diversity. Twelve-lead
electrocardiograms~\cite{wagner2020ptbxl,goldberger2000} from a healthy and a myocardial-infarction cohort show CWMMSE correcting a misranking by
averaged entropy. Ensembles of satellite orbits and breakup debris~\cite{vallado2006,celestrak} serve as a null control: they are
geometrically diverse yet individually predictable, so a diversity-only measure would misrank them. The
remaining seven systems form a cross-domain survey that tests whether the same decomposition recurs without
per-system tuning: fluid turbulence~\cite{weine2024cmrsim}, tropical-cyclone tracks~\cite{knapp2010},
maritime vessel traffic~\cite{marinecadastre}, a broadband record
of the 2011 Tohoku earthquake~\cite{beyreuther2010}, variable-star light curves~\cite{dau2019ucr},
neurodegenerative gait~\cite{hausdorff2000}, and a continental GNSS
network~\cite{blewitt2018ngl}.

The section is organised around these roles. We first validate CWMMSE on deterministic chaos, then use the
ocean drifters to separate the two ingredients and the electrocardiogram cohorts to show that averaging can
rank a diseased population as the more complex while CWMMSE reverses it. Two controls follow, the cardiac
contrast that a mean-only measure fails and the orbital null that a diversity-only measure fails,
establishing that both ingredients are necessary and motivating the weighted-entropy form. We then survey
the seven cross-domain systems, place all eleven in the complexity plane, and test robustness to the
clustering resolution, the trajectory distance and the linkage rule.

Throughout, trajectories are clustered by Ward linkage~\cite{ward1963} on the default distance and CWMMSE uses
equation~\eqref{eq:cwmmse}. Reported intervals are $2.5$--$97.5$ percentile subsample stability ranges
(size $0.8M$, sampled without replacement). Because the subsamples overlap heavily, they summarise the
stability of the estimate and understate the true sampling variance; we therefore report them as stability
ranges rather than confidence intervals. We prefer a single such convention across all functionals to the
closed-form DeLong AUC variance~\cite{delong1988}. A case-resampling bootstrap over trajectories is
inappropriate here, because duplicate draws lie at zero distance and distort the clustering.

\subsection*{Validation against known dynamics}
CWMMSE recovers the complexity ordering that the dynamics dictate. For several regimes of the
Lorenz~\cite{lorenz1963} and R\"ossler~\cite{rossler1976} systems we integrate
an ensemble of $100$ trajectories from perturbed initial conditions and compute CWMMSE
(Table~\ref{tab:validation}, Figure~\ref{fig:validation}). Within each system CWMMSE is largest in the
chaotic regime, Lorenz $\rho=28$ ($S_{\mathrm{CWMMSE}}=0.95$) and R\"ossler in its chaotic regime ($0.12$), and falls
to near zero at the fixed-point and low-period regimes, recovering the ordering expected from the
dynamics. In the chaotic regime each ensemble becomes maximally diverse ($\hat H\to1$, every trajectory its own
pattern), and the Lorenz ensemble is individually complex as well. A spectrum-matched surrogate test confirms that
this responds to nonlinear structure rather than to the linear spectrum. We replace each trajectory by an
iterative amplitude-adjusted Fourier-transform surrogate~\cite{theiler1992,schreiber1996}, which preserves
the power spectrum and amplitude distribution but randomises phases. This raises the chaotic Lorenz CWMMSE of a
60-trajectory subsample from $0.84$ to $8.3$ (mean across $999$ surrogate ensembles), below every one of them (one-sided $p=0.001$). That
the surrogate scores far higher is itself informative: sample entropy rises with per-trajectory
unpredictability, so noise of the same spectrum is more irregular per trajectory than deterministic motion
and scores above it. CWMMSE thus distinguishes nonlinear structure from the linear spectrum, and within a family of comparable
stochasticity it places chaos above order.

\begin{table}[tb]
\centering\footnotesize
\caption{CWMMSE on Lorenz and R\"ossler ensembles ($100$ trajectories each). $k$, number of clusters;
$\hat H$, normalised cluster entropy; $\bar S$, mean MMSE. CWMMSE peaks in the chaotic regime of each
system.}\label{tab:validation}
\begin{tabular}{llrrrr}
\toprule
System & Regime & $k$ & $\hat H$ & $\bar S$ & CWMMSE\\
\midrule
Lorenz   & fixed point      & 1   & 0.00 & 0.00 & 0.00\\
Lorenz   & transition       & 2   & 0.09 & 0.00 & 0.00\\
Lorenz   & chaotic          & 100 & 1.00 & 0.21 & \textbf{0.95}\\
Lorenz   & periodic         & 66  & 0.87 & 0.14 & 0.57\\
R\"ossler & periodic        & 3   & 0.24 & 0.03 & 0.03\\
R\"ossler & quasi-periodic  & 32  & 0.54 & 0.02 & 0.04\\
R\"ossler & chaotic         & 89  & 0.97 & 0.03 & \textbf{0.12}\\
\bottomrule
\end{tabular}
\end{table}

\begin{figure}[tb]
\centering
\includegraphics[width=\linewidth]{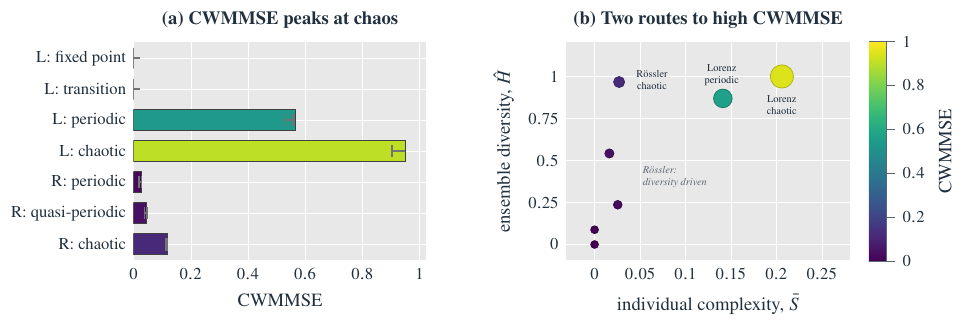}
\caption{Validation on deterministic chaos, where the correct complexity ordering is known in advance. Ensembles of $100$ trajectories are drawn from the Lorenz (L) and R\"ossler (R) systems in regimes ranging from a fixed point through periodic motion to fully developed chaos. (a) CWMMSE of each regime; bar height encodes CWMMSE, with a colour scale from dark (low CWMMSE) to bright (high CWMMSE), and the whiskers are $2.5$--$97.5$ percentile subsample intervals. The value is near zero for the fixed-point and periodic regimes and peaks in the chaotic regime of each system, recovering the expected ordering. (b) The same ensembles in the plane of the two ingredients, individual complexity $\bar S$ (mean within-trajectory MMSE) on the horizontal axis against ensemble diversity $\hat H=H(C)/\ln M$ (normalised cluster entropy) on the vertical, with marker area proportional to CWMMSE and the same dark-to-bright colour scale. CWMMSE is largest where both ingredients are large, at the chaotic Lorenz regime (top right); regimes that are diverse but individually simple (chaotic R\"ossler, top left) or low in either ingredient carry small CWMMSE.}
\label{fig:validation}
\end{figure}

A controlled construction isolates the mechanism. With $k$ equal clusters each of uniform complexity $S$
one finds $S_{\mathrm{CWMMSE}}=S\ln k$ exactly: the mean MMSE then depends only on $S$, the cluster
entropy only on $k$, and CWMMSE on both. Hand-built ensembles (ESM) confirm that CWMMSE, unlike either
ingredient alone, singles out a population that is simultaneously diverse and individually complex.

\subsection*{Individual complexity versus ensemble diversity: ocean drifters}
With the measure validated on dynamics we control, we turn to field data, where the two ingredients diverge. From hourly surface-velocity records of the NOAA Global Drifter
Program~\cite{elipot2016drifter}, we form ensembles for an energetic western-boundary-current region (Gulf
Stream) and a quiescent gyre interior (Sargasso Sea), using $120$-hour windows of the two velocity
components. The Gulf Stream is more complex \emph{individually}, its drifters caught in eddies and
reversing direction, with the higher mean MMSE ($\bar S=0.87$ against $0.76$); yet it is the simpler
\emph{system}, with the lower CWMMSE ($2.08$ against $2.86$). The energetic flow entrains its drifters
into a few shared coherent structures ($\hat H=0.46$), whereas the quiescent gyre supports a far more
diverse population of paths ($\hat H=0.66$) (Figure~\ref{fig:decomp}, top). Averaged entropy and CWMMSE
thus rank the two regions oppositely, with the difference residing entirely in the diversity term; the
ranking is in the stated direction in $99.9\%$ of subsamples. The mechanism is physical. The western
boundary current organises the flow into a few coherent meanders and rings that many drifters follow in
common, so individually tortuous paths are collectively redundant. The gyre interior lacks such organising
structures, and each drifter traces its own slow path, leaving the population as a whole more varied. The
two regions are thus ranked oppositely by individual complexity and by system complexity, the divergence
the decomposition is designed to expose.

\begin{figure}[tb]
\centering
\includegraphics[width=\linewidth]{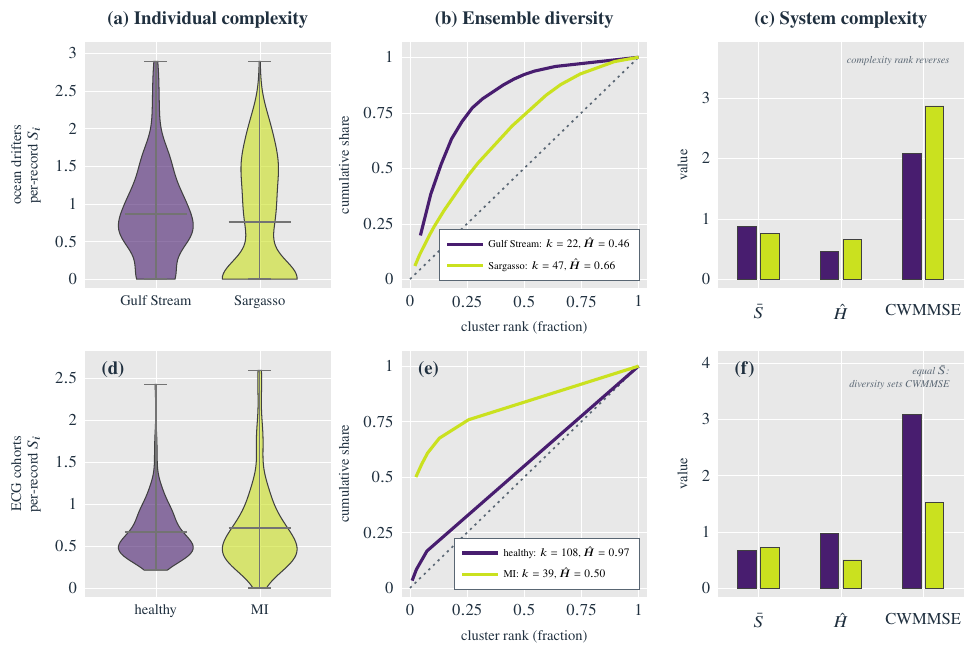}
\caption{The decomposition in field data: the individually more complex group can be the less complex \emph{system}. Two contrasts are shown: ocean-drifter regions (top) and ECG cohorts (bottom), with the two groups of each contrast drawn in two distinct colours, violet (Gulf Stream; NORM) and yellow-green (Sargasso; MI). (a,d) Violin distributions of the per-record within-trajectory complexity $S(T_j)$: the Gulf Stream is individually a little more complex than the Sargasso, whereas the healthy (NORM) and infarction (MI) hearts are nearly equal. (b,e) Cumulative share of members against cluster rank (clusters ordered largest first), annotated with the cluster count $k$ and the normalised cluster entropy $\hat H=H(C)/\ln M$; a steeper curve means a few clusters hold most members, that is a more redundant, less diverse ensemble (the Gulf Stream, $\hat H=0.46$; the MI cohort, $\hat H=0.50$). (c,f) The two ingredients $\bar S$ and $\hat H$ and the resulting CWMMSE side by side. In each row the diversity term overturns the individual-complexity ranking, so the more diverse group has the larger CWMMSE (Sargasso $2.86$ against Gulf Stream $2.08$; healthy above MI), although an averaged complexity would favour the individually more complex Gulf Stream and would rank the diseased MI hearts above the healthy ones.}
\label{fig:decomp}
\end{figure}

\subsection*{A sign reversal of averaging: ECG cohorts}
On clinical electrocardiograms the conventional averaged
entropy ranks diseased hearts as more complex than healthy ones; any diversity-aware reading reverses that
verdict, and the reversal exposes population structure that averaging hides. On
12-lead electrocardiograms from the PTB-XL database~\cite{wagner2020ptbxl,goldberger2000}, using leads I,
II and V2 with each 10-second record as one multivariate trajectory, the healthy (NORM) and
myocardial-infarction (MI) cohorts are nearly indistinguishable at the level of the individual: their mean
MMSE is $0.665$ and $0.716$, the diseased cohort marginally \emph{higher}. As a population, however, the healthy cohort is far more diverse ($\hat H=0.97$, $k=108$)
than the diseased one ($\hat H=0.50$, $k=39$), so CWMMSE separates them by roughly a factor of two
($3.08$ against $1.51$; Figure~\ref{fig:decomp}, bottom). This complements the classical loss of physiological
complexity with disease and ageing~\cite{goldberger2002,costa2002,costa2005}: the per-record complexity here
runs the other way, and the cohort-level signature is instead a
contraction of ensemble diversity invisible to single-series measures, which we read as lower system
complexity. The reading is clinically intelligible: a healthy population samples a wide repertoire of normal
cardiac morphologies, whereas infarction imposes a few stereotyped abnormal patterns that contract the
population's diversity even as each diseased record remains locally irregular. The contraction is a
property of the population, not of any individual record, which is exactly why a measure that averages
over records cannot detect it.

A subsampling benchmark quantifies the consequence (Table~\ref{tab:ecg}, Figure~\ref{fig:ecg}). Drawing
size-50 subsamples from each cohort and asking how reliably a functional ranks a healthy subsample above a
diseased one, we find the averaged MMSE attains an area under the curve (AUC) of only $0.23$, below the chance
value of $0.5$: it ranks the cohorts in the wrong direction, assigning higher per-record complexity to the diseased cohort.
CWMMSE separates the cohorts with AUC $0.92$ and a large effect size ($|d|=2.1$, the cohort
means differing by about two standard deviations). An
external functional-diversity index, Rao's quadratic entropy in its uniform-abundance form (the mean pairwise distance), also
misranks the cohorts (AUC $0.004$), because raw pairwise diversity is dominated by amplitude rather than by
dynamical complexity. Two interpretations of this contrast require qualification. First, the cluster entropy alone separates
the cohorts even more sharply than CWMMSE (AUC $0.995$), which might suggest that diversity alone suffices;
the orbital null below shows that it does not. Second, the
separable product $H(C)\,\bar S$ matches CWMMSE on this task (AUC $0.91$ against $0.92$), so the cardiac
separation establishes the value of a two-ingredient measure but cannot by itself favour the
weighted-entropy form over the product. The next two results resolve both points: the orbital null, and the
complexity--diversity coupling of equation~\eqref{eq:coupling}.

\begin{table}[tb]
\centering\footnotesize
\caption{Separating healthy from MI ECG subsamples (clustering resolution $0.3$). AUC is the probability
that the functional ranks a healthy subsample above an MI one; $0.5$ is no separation and below $0.5$ is
the wrong direction. $|d|$, Cohen's effect size; brackets, subsample stability ranges (not confidence intervals).}\label{tab:ecg}
\begin{tabular}{lrr}
\toprule
Functional & AUC & $|d|$\\
\midrule
mean MMSE $\bar S$            & 0.23~[0.20, 0.26] & 1.06\\
cluster entropy $H(C)$        & 0.995~[0.99, 1.00] & 2.88\\
$H(C)\,\bar S$ (product)      & 0.91~[0.89, 0.93] & 2.09\\
Rao $Q$ (pairwise diversity)  & 0.004~[0.00, 0.01] & 3.65\\
CWMMSE                        & \textbf{0.92~[0.89, 0.94]} & \textbf{2.14}\\
\bottomrule
\end{tabular}
\end{table}

\begin{figure}[tb]
\centering
\includegraphics[width=\linewidth]{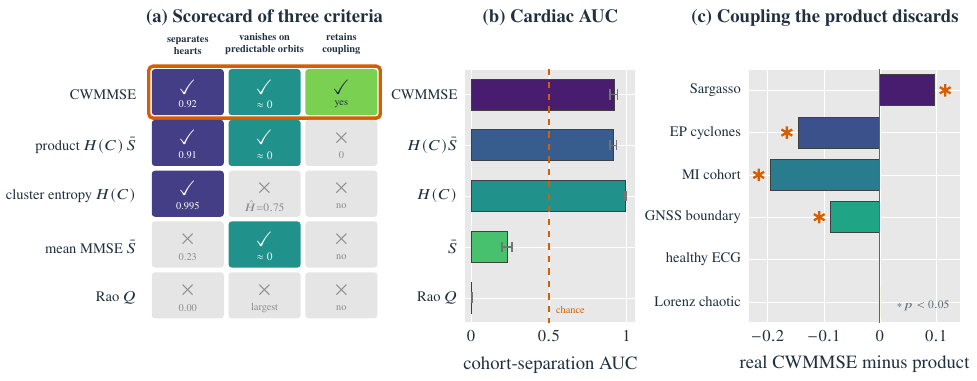}
\caption{The two controls and the complexity-diversity coupling, summarised as a scorecard. Five ensemble functionals are scored: the mean multivariate multiscale sample entropy (MMSE) $\bar S$, the cluster entropy $H(C)$, their separable product $H(C)\bar S$, the pairwise-diversity Rao $Q$, and CWMMSE. (a) Each functional is checked against three criteria, laid out as a grid with passing cells coloured by criterion and the value shown. The criteria are: whether it separates the PTB-XL healthy (NORM) and myocardial-infarction (MI) cohorts (cohort-separation area under the curve, AUC, above the chance value of $0.5$); whether it vanishes on a maximally spread but dynamically trivial ensemble of predictable orbits (the orbital null); and whether it retains the complexity-diversity coupling of equation~\eqref{eq:coupling}. Only CWMMSE passes all three (highlighted row): the single-ingredient measures each fail a control, and the separable product passes both controls but discards the coupling by construction. (b) The cardiac cohort-separation AUC of each functional, with $95\%$ intervals, at clustering resolution $0.3$ (dashed line, chance). The diversity-only $H(C)$ separates the hearts most sharply, yet fails the orbital control in (a). (c) The coupling itself, the real CWMMSE minus the separable product, across systems. It is significant (starred, $p<0.05$) in several real ensembles and negligible only where cluster complexities are homogeneous. This is the information that the product, and every single-ingredient measure, discards.}
\label{fig:ecg}
\end{figure}

\subsection*{A complexity control, and why both ingredients are needed}
Diversity alone is not complexity. For ensembles of low-Earth-orbit objects propagated from two-line
elements~\cite{vallado2006,celestrak} (the standard catalogue records of each object's orbit), comprising an intact
constellation and two breakup debris clouds, the
normalised cluster entropy is the \emph{highest} of any system we study ($\hat H\approx0.75$) and the
mean pairwise (Rao) diversity is likewise the largest, since the objects populate many distinct orbital
planes; yet individual orbital motion is essentially perfectly predictable, so the per-trajectory
complexity is near zero and $S_{\mathrm{CWMMSE}}\approx0$ throughout (Figure~\ref{fig:ecg}a, with the
quantitative orbital null given in the ESM). A measure of
diversity alone, whether the cluster entropy or Rao's quadratic entropy, would rank these objects the most
complex systems studied; CWMMSE correctly returns near zero. The breakup debris clouds score marginally
above the uniform constellation, because fragmentation imparts small, slightly irregular perturbations to
otherwise regular orbits; CWMMSE registers this faint signal, whereas a saturated diversity measure cannot.

Together, the two controls establish what a measure of system complexity must contain. The mean MMSE fails the cardiac contrast; the
cluster entropy and Rao's quadratic entropy fail the orbital null. The functionals that pass both are
exactly those built from both ingredients: the separable product and CWMMSE. CWMMSE is preferred over the
product because it is a single, established information-theoretic object (equation~\eqref{eq:weighted})
rather than an ad hoc product, and because it retains the complexity-diversity coupling of
equation~\eqref{eq:coupling} that the product discards; the consistency of Theorem~\ref{thm:weighted} holds
equally for the product and so does not by itself favour either. The coupling is tested by a permutation null that holds the clustering fixed
and shuffles the per-trajectory complexities, breaking any link between which patterns are complex and
which are common, so that its mean equals the separable product. This null shows the coupling to be
significant under Benjamini--Hochberg correction in several ensembles. It is positive for the
Sargasso ensemble ($+0.10$, corrected $p=0.002$), where the more complex patterns are
also the more prevalent, and negative for the East-Pacific cyclones ($-0.15$, $p<0.001$) and the MI cohort
($-0.20$, $p=0.04$). The GNSS plate-boundary coupling ($-0.09$) does not survive correction ($p=0.09$) and
is reported as suggestive only. The coupling modulates the magnitude
of these contrasts without overturning any featured ranking, so it strengthens the contrasts rather than
being necessary to establish them.

\subsection*{Cross-domain synthesis}
The same two-ingredient decomposition holds across the seven cross-domain systems with no per-system tuning
(Table~\ref{tab:cross}; ESM). Turbulent particle ensembles~\cite{weine2024cmrsim} are far more complex than
laminar ones (CWMMSE $2.05$ against $0.15$). Tropical cyclone tracks~\cite{knapp2010} order East Pacific above North
Atlantic above West Pacific, the most erratic tracks scoring highest, and manoeuvring maritime vessels
are more complex as a system than lane-following cargo and tankers. At a broadband seismometer, system
complexity is highest in the pre-event noise and collapses during the coherent wave train of the 2011
Tohoku earthquake~\cite{beyreuther2010}, which the measure registers as an abrupt drop in system complexity. Three further
domains extend the reach: variable-star light curves~\cite{dau2019ucr} separate pulsation classes,
gait-force recordings~\cite{hausdorff2000} separate neurodegenerative cohorts from healthy controls, and
a continental GNSS network~\cite{blewitt2018ngl} separates active plate boundaries from stable interiors.
Placed in the plane of individual complexity against ensemble diversity (Figure~\ref{fig:plane}), the
eleven systems span its full
extent, from high in both (developed chaos, turbulence), through the diversity-only corner (orbits), to
the complexity-loss cases (diseased cohorts, the earthquake), with CWMMSE largest where both coordinates
are high and near zero whenever either collapses.

\begin{table}[tb]
\centering\footnotesize
\caption{The eleven applications, with the scientific domain and data source of each system. Values are
representative CWMMSE for the stated contrast; full settings, sample sizes and intervals are in the
ESM.}\label{tab:cross}
\begin{tabular}{lllll}
\toprule
System & Domain & Source & Contrast & CWMMSE\\
\midrule
Lorenz/R\"ossler & chaos        & \cite{lorenz1963,rossler1976} & regime               & peaks at chaos\\
Turbulence       & fluids       & \cite{weine2024cmrsim}        & turbulent / laminar  & 2.05 / 0.15\\
Drifters         & ocean        & \cite{elipot2016drifter}      & Sargasso / Gulf St.  & 2.86 / 2.08\\
Cyclones         & atmosphere   & \cite{knapp2010}              & EP / NA / WP         & 2.09 / 1.35 / 0.79\\
AIS vessels      & maritime     & \cite{marinecadastre}         & working / cargo      & 0.61 / 0.24\\
Orbits           & astronautics & \cite{celestrak}              & debris / fleet       & 0.02 / 0.00\\
ECG              & physiology   & \cite{wagner2020ptbxl}        & healthy / MI         & 3.08 / 1.51\\
Seismic          & geophysics   & \cite{beyreuther2010}         & pre / wave train     & 0.34 / 0.03\\
Variable stars   & astronomy    & \cite{dau2019ucr}             & eclipsing / Cepheid  & 0.08 / 0.05\\
Gait             & clinical     & \cite{hausdorff2000}          & ALS / control        & 0.22 / 0.15\\
GNSS             & geodesy      & \cite{blewitt2018ngl}         & boundary / interior  & 0.42 / 0.06\\
\bottomrule
\end{tabular}
\end{table}

\begin{figure}[tb]
\centering
\includegraphics[width=\linewidth]{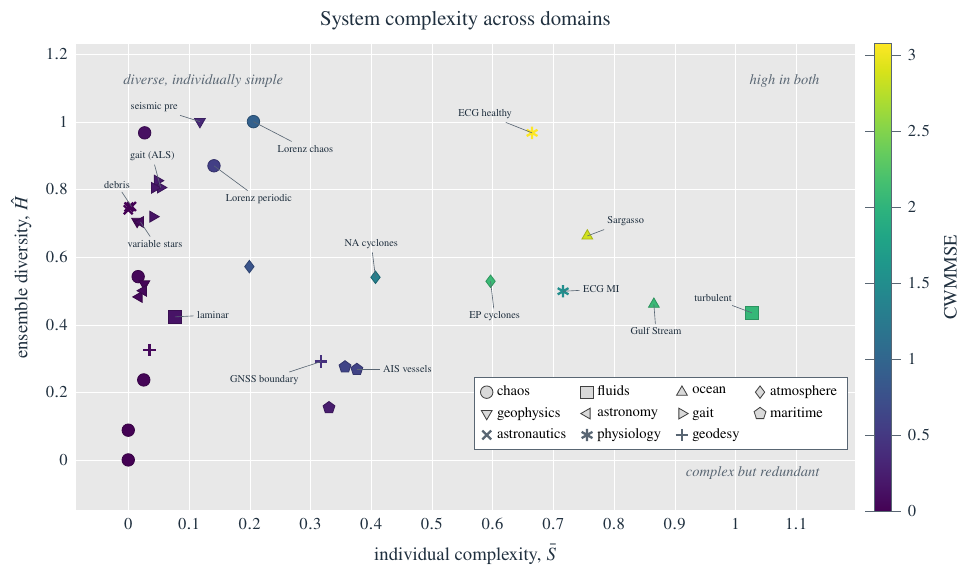}
\caption{All eleven application domains in the complexity plane. Each marker is one group (a dynamical regime, geographic region, clinical cohort, vessel or orbital class, basin or time window): the horizontal axis is individual complexity $\bar S$ (mean within-trajectory MMSE), the vertical axis is ensemble diversity $\hat H=H(C)/\ln M$ (normalised cluster entropy), marker colour encodes CWMMSE on a scale from dark (low CWMMSE) to bright (high CWMMSE), and marker shape encodes the scientific domain (legend). The two ingredients are nearly independent across the corpus, so the groups fill the plane rather than lying on a line. CWMMSE (colour) brightens toward the upper right and is largest for ensembles substantial in both ingredients (the healthy ECG cohort, the turbulent and chaotic ensembles, the Sargasso drifters); it is near zero wherever either ingredient collapses, as for the high-diversity but dynamically predictable orbital-debris and pre-event seismic ensembles at the top left, and as it would be for an individually complex but behaviourally redundant ensemble at the lower right. The plane is comparative rather than absolute, since $\bar S$ and $\hat H$ depend on the per-domain analysis parameters listed in the Electronic Supplementary Material; in particular $\hat H$ depends on the ensemble size $M$, so cross-domain positions are read qualitatively and are not comparable in absolute terms unless $\hat H$ is recomputed at a common $M$.}
\label{fig:plane}
\end{figure}

\subsection*{Robustness}
Because the clustering granularity is a free resolution and the distance carries a scale, we check that
the contrasts are not artefacts of either. Sweeping the relative threshold over a broad range leaves every
contrast's direction unchanged on a wide plateau, and replacing the amplitude-dependent distance by a
scale- and length-invariant normalisation preserves every featured contrast (Sargasso above Gulf Stream,
healthy above MI, the chaotic and basin orderings) with sign probability at least $0.99$. An
amplitude-free shape distance instead drives the clustering to near-triviality and is reported as a
documented sensitivity. The contrasts are equally insensitive to the agglomeration rule, holding under
Ward, average and complete linkage. The cardiac contrast is robust across the tested grid of
sample-entropy tolerance and embedding; the more marginal drifter contrast holds at the default and
larger tolerances but weakens at the smallest, a sensitivity we report rather than tune away (ESM).
Off-the-shelf parameter-free cluster-number rules are inappropriate here: they
select two or three coarse clusters and flip the cardiac contrast, because CWMMSE measures fine-grained
ensemble diversity, not the single dominant split. We therefore treat the resolution as explicit
and report the full sweep (ESM). Across every check, the direction of every featured contrast is
preserved, so the decomposition, not the parameters, drives the results.

\section{Discussion}\label{sec:discussion}

Taken together, these results cohere into a single account. We have introduced CWMMSE, a measure of the complexity of a whole ensemble of trajectories, and shown that it
is exactly the weighted entropy of Belis and Guia\c{s}u~\cite{belis1968,guiasu1971} applied to the population's
distribution over behavioural patterns, with each pattern's dynamical complexity as its utility. The
prediction we set out, that system complexity decomposes into individual complexity and ensemble diversity
and that the two diverge in practice, is confirmed across all eleven systems: a region of ocean can be
individually complex yet collectively redundant, a debris cloud maximally diverse yet individually
trivial, and a diseased cohort can match a healthy one in per-record complexity yet be far less diverse as
a population.

This independence is why CWMMSE captures something its constituents cannot, and why both ingredients are
required. A measure built on either ingredient alone fails where the other dominates, as the cardiac and
orbital controls show (\S\ref{sec:results}): averaged entropy misranks the cohorts, and a diversity index
such as Rao's quadratic entropy~\cite{rao1982,ricottaszeidl2006,leinstercobbold2012} ranks perfectly
predictable orbits as the most complex system studied. This is the central distinction from
functional-diversity theory: those indices weight a population's
spread by pairwise dissimilarities or trait differences, so a maximally spread but dynamically trivial
ensemble is maximally diverse to them; CWMMSE weights instead by intrinsic dynamical complexity, and so
returns near zero. Relatedly, we report CWMMSE as a dimensionless, complexity-weighted entropy rather than as an ordinary Shannon
entropy in nats, and do not convert it to an effective number of patterns~\cite{jost2006}, giving the
normalised cluster entropy as its diversity-only companion. Both the separable product and CWMMSE combine the two ingredients
and pass both controls, but they are not equivalent: equation~\eqref{eq:coupling} expresses their difference
as a population-weighted covariance between each pattern's complexity and its surprisal, which is significant
in several systems, and CWMMSE is the
single information-theoretic object, with the asymptotic meaning of Theorem~\ref{thm:weighted}, that
retains it.

The result connects two strands of theory. The first is weighted entropy~\cite{belis1968,guiasu1971,
shannon1948,cover2006}, of which CWMMSE is the instance whose utility is the outcome's own dynamical
complexity. The second is the loss of complexity with disease and ageing~\cite{goldberger2002,costa2002,
costa2005}, established for individual physiological signals; our cardiac result lifts that principle from
the individual to the cohort, where its signature is a contraction of ensemble diversity. More broadly,
wherever data arrive as a population of trajectories, in ecology, clinical screening, transport and
infrastructure, or space situational awareness, the decomposition offers a vocabulary for distinguishing
systems that are complex because their members are intricate from systems that are complex because their
members are varied.

We recommend a three-part practice. We fix the clustering resolution within the robust plateau and
report it; we present CWMMSE alongside its two ingredients, so that a contrast can be read as driven by
individual complexity, by ensemble diversity, or by both; and we place each system in the
complexity-diversity plane (Figure~\ref{fig:plane}), where qualitatively different regimes occupy distinct
corners. Reported this way, CWMMSE is a comparative instrument that
ranks ensembles within a study and diagnoses which ingredient drives a difference, rather than returning
an absolute complexity on a universal scale; this is by design, and is shared with the diversity indices
and effect-size measures already standard in the fields it targets. The accompanying open implementation makes this routine: a
single function returns CWMMSE with its two ingredients and a subsample interval for any ensemble supplied
as a list of multivariate trajectories.

Three aspects define the scope of the measure and mark its natural extensions. First, the clustering
resolution is a genuine parameter of the method rather than a nuisance: we report it explicitly and show
that every featured contrast is stable across a broad plateau (\S\ref{sec:results}). The parameter-free
cluster-number rules we tested select the single dominant split, which does not track the fine-grained
diversity CWMMSE is designed to measure, so we set the resolution deliberately and report the full sweep.
The resolution fixes the measure to within-study comparison, as does the ensemble-size dependence of the
normalised cluster entropy $\hat H=H(C)/\ln M$, which for a fixed finite partition tends to zero as $M$
grows; we report the ensemble size and a bias-corrected cluster entropy throughout (ESM), and a principled,
data-driven rule for the resolution would extend it to comparison across studies. Second, the within-trajectory weight
is the sample entropy, and so measures irregularity, inheriting its sensitivities to the tolerance and
embedding and a dependence on the trajectory dissimilarity, whose amplitude dependence we bound by the
scale-invariant check. Because sample entropy rates unstructured noise as maximally irregular, a
spectrum-matched noise ensemble scores above deterministic chaos (\S\ref{sec:results}); CWMMSE therefore
ranks systems of comparable stochasticity rather than separating structure from noise. A per-trajectory
weight that explicitly discounts unstructured noise, in place of the sample entropy used here, would
sharpen this separation. Third, the eleven applications are deliberately broad rather than
exhaustive: their purpose is to show that one decomposition recurs across physically unrelated systems
without per-system tuning, and turning any single case, most naturally the physiological one, into a
definitive study is an inviting direction that would require larger cohorts and comparison with established
markers.

These directions form a clear agenda. An axiomatic characterisation of the weighted-entropy form,
following the axiomatic treatment of diversity measures~\cite{leinster2021}, would establish in what sense
it is the unique combination of the two ingredients; a data-driven rule for the resolution would remove the
last free parameter; and a uniform multiscale and multivariate treatment across domains would sharpen the
comparisons further.

\section{Conclusions}\label{sec:conclusions}

Clustering-weighted multivariate multiscale sample entropy measures the complexity of a population of
trajectories, rather than of any single member. It is a weighted entropy of the population's distribution
over behavioural patterns, in which each pattern is weighted by its own dynamical complexity. In this form
it separates two properties that a simple average conflates: how irregular the members are, and how varied
the population is. Across eleven systems spanning the physical, environmental, engineering and biomedical
sciences, these two properties vary largely independently, so that individual and collective complexity can
point in opposite directions. A measure sensitive to only one of them is therefore misled, whereas a measure
that carries both, together with the coupling between them, tracks the complexity of the system as a whole.
The broader lesson is that the complexity of an ensemble is not the average of its parts: read as a system
in its own right, a population reveals structure, such as the contraction of a cohort's diversity with
disease, that averaging conceals. As observations across the sciences increasingly arrive as populations of
trajectories, measuring the complexity of the whole, rather than averaging the complexity of its members, is
both feasible and, we argue, necessary.

\enlargethispage{12pt}

\ethics{This study analyses only previously published, openly available datasets. The electrocardiogram
data (PTB-XL) and the neurodegenerative-disease gait recordings are de-identified and were collected under
the original studies' ethical approvals and informed
consent~\cite{wagner2020ptbxl,goldberger2000,hausdorff2000}. No new experiments on humans or animals were
performed.}

\dataccess{All code, with one-command reproduction, is openly available at
\url{https://github.com/cisgroup/cwmmse} (DOI and Zenodo will be added on acceptance). Supporting
analyses, per-dataset parameters and full results are in the electronic supplementary material (ESM).
Every dataset is openly available: the NOAA Global
Drifter Program~\cite{elipot2016drifter}, NOAA IBTrACS~\cite{knapp2010}, Marine
Cadastre AIS~\cite{marinecadastre}, CelesTrak two-line elements~\cite{celestrak}, PhysioNet
PTB-XL~\cite{wagner2020ptbxl,goldberger2000}, IRIS and EarthScope FDSN~\cite{beyreuther2010}, the UCR
StarLightCurves archive~\cite{dau2019ucr}, the PhysioNet gait-in-neurodegenerative-disease
database~\cite{hausdorff2000}, and the Nevada Geodetic Laboratory GNSS series~\cite{blewitt2018ngl}.
Per-dataset sources, access dates and parameters are tabulated
in the electronic supplementary material.}

\aucontribute{C.T.: conceptualisation, software, investigation, writing of the original draft. J.H.: conceptualisation, methodology, supervision, writing, review and editing. Both authors gave final approval for publication.}

\aiuse{During the preparation of this manuscript, the authors used Claude (Anthropic, Opus 4.8) to check spelling and grammar and to correct typographical errors. The tool was not used to generate scientific content, data, results, or interpretations. After using this tool, the authors reviewed and edited the text as needed and take full responsibility for the content of the publication.}

\competing{We declare we have no competing interests.}

\funding{This work has been supported by a grant from the Fund for Energy Research with Corporate Partners administered by the Andlinger Center for Energy and the Environment at Princeton University.}

\ack{We thank the data providers listed under Data accessibility.}

\printbibliography
\end{refsection}

\clearpage
\setcounter{section}{0}\renewcommand{\thesection}{S\arabic{section}}
\setcounter{figure}{0}\renewcommand{\thefigure}{S\arabic{figure}}
\setcounter{table}{0}\renewcommand{\thetable}{S\arabic{table}}
\setcounter{equation}{0}\renewcommand{\theequation}{S\arabic{equation}}

\begin{center}
  {\LARGE\bfseries Supplementary Information}\\[1.5ex]
  {\large Quantifying the complexity of trajectory ensembles with clustering-weighted multivariate multiscale sample entropy}
\end{center}
\bigskip

\begin{refsection}

This document supports the main article. Throughout, CWMMSE denotes the clustering-weighted multivariate
multiscale sample entropy of an ensemble, $S_{\mathrm{CWMMSE}}=\sum_i(-p_i\ln p_i)S_i$, which combines
each pattern's within-trajectory complexity $S_i$ with its share $-p_i\ln p_i$ of the population's
diversity. Each section below supplies the full evidence behind a claim of the main text: it collects the
symbols used throughout (\S\ref{s:notation}), documents the
full multivariate multiscale sample entropy (MMSE) algorithm and parameters (\S\ref{s:alg}), a proof of
the consistency theorem (\S\ref{s:proof}),
the functional family and the justification of the weighting (\S\ref{s:family}), the full
electrocardiogram (ECG) benchmark (\S\ref{s:ecg}), the clustering-distance and resolution robustness
(\S\ref{s:robust}), the surrogate and null tests (\S\ref{s:surr}), the eleven systems in detail with
illustrative figures (\S\ref{s:systems}), and data and code availability (\S\ref{s:data}). All numbers
are regenerated from open data by the released software package. Reported intervals are $2.5$--$97.5$
percentile subsample stability ranges (size $0.8M$, sampled without replacement); they understate the true sampling
variance, so we report them as stability ranges rather than confidence intervals.

\section{Notation}\label{s:notation}

Table~\ref{tab:notation} collects the symbols used in the main text and this supplement. Trajectories are
indexed by $j$, clusters (behavioural patterns) by $i$, and channels within a trajectory by $a$.

\begin{table}[h]
\centering\small
\caption{Symbols used in the main text and the electronic supplementary material.}\label{tab:notation}
\begin{tabular}{ll}
\toprule
Symbol & Meaning\\
\midrule
\multicolumn{2}{l}{\emph{Ensemble and trajectories}}\\
$\mathcal{D}=\{T_1,\dots,T_M\}$ & ensemble of $M$ trajectories\\
$M$ & number of trajectories (ensemble size)\\
$T_j$ & trajectory $j$ \ ($j=1,\dots,M$)\\
$N_j$ & length (number of samples) of trajectory $j$\\
$c$ & number of channels per trajectory\\
$a$ & channel index \ ($a=1,\dots,c$)\\
$x^{(a)}_l$ & sample $l$ of channel $a$\\
$\delta_{jj'}$ & dissimilarity between trajectories $j$ and $j'$\\
\midrule
\multicolumn{2}{l}{\emph{Within-trajectory complexity (MMSE)}}\\
$s$ & temporal scale\\
$y^{(a),s}_t$ & coarse-grained channel $a$ at scale $s$ (time index $t$)\\
$\mathbf{m}=(m_1,\dots,m_c)$ & per-channel embedding dimensions\\
$\boldsymbol{\tau}=(\tau_1,\dots,\tau_c)$ & per-channel embedding delays\\
$d=\sum_a m_a$ & composite embedding dimension\\
$r$ & matching tolerance (Chebyshev metric)\\
$V_d$ & composite delay vector of dimension $d$\\
$B_d(r)$ & fraction of vector pairs within $r$ at dimension $d$\\
$S^{(s)}(T_j)$ & scale-$s$ sample entropy of trajectory $j$\\
$S(T_j)$ & within-trajectory complexity (single-scale or scale-summed)\\
\midrule
\multicolumn{2}{l}{\emph{Clustering and diversity}}\\
$C_i$ & cluster (behavioural pattern) $i$\\
$k$ & number of clusters\\
$i$ & cluster index \ ($i=1,\dots,k$)\\
$p_i=|C_i|/M$ & proportion of trajectories in cluster $i$\\
$S_i$ & mean within-trajectory complexity of cluster $i$\\
$h_i=-p_i\ln p_i$ & surprisal weight of cluster $i$\\
$H(C)=-\sum_i p_i\ln p_i$ & Shannon cluster entropy\\
$\hat H=H(C)/\ln M$ & normalised cluster entropy (ensemble diversity)\\
$\hat p_i,\ \hat S_i$ & empirical estimates of $p_i,\ S_i$\\
$M_i=M\hat p_i$ & membership count of cluster $i$\\
\midrule
\multicolumn{2}{l}{\emph{Ensemble functionals}}\\
$\bar S=\sum_i p_i S_i$ & mean MMSE (individual complexity)\\
$S_{\mathrm{CWMMSE}}=\sum_i(-p_i\ln p_i)\,S_i$ & clustering-weighted MMSE\\
$H^w=-\sum_i w_i\,p_i\ln p_i$ & weighted entropy ($w_i$, per-outcome utility)\\
$Q=\tfrac{1}{M^2}\sum_{j,j'}\delta_{jj'}$ & Rao's quadratic entropy (uniform-abundance form)\\
\bottomrule
\end{tabular}
\end{table}

\section{The MMSE algorithm and parameters}\label{s:alg}

The within-trajectory complexity is the multivariate multiscale sample entropy~\cite{ahmed2011,costa2002}
built on sample entropy~\cite{richman2000}. For a multivariate trajectory the channels are coarse-grained
at scale $s$ (main text equation~(2.1)), per-channel standardised, embedded into composite delay vectors of
dimension $d=\sum_a m_a$, and matched under the Chebyshev metric within a tolerance $r$; the sample entropy
is the negative log ratio of neighbour counts at successive embeddings, both counted over the template
indices admissible at the larger embedding so the ratio lies in $[0,1]$ (Algorithm~\ref{alg:mmse}). We
count neighbours with a $k$-d tree, giving roughly order $N\log N$ time at the low embedding dimensions
used here (composite dimension $d=\sum_a m_a\le 6$) rather than the order $N^2$ of an explicit distance
matrix. The extended step lengthens every channel's embedding by one sample at once, raising the composite
dimension from $d$ to $d+c$ (an increase of $c$, not $1$); this multivariate convention follows our released
implementation, coincides with the construction of Ahmed and Mandic~\cite{ahmed2011} for a single channel,
and for $c>1$ is a distinct construction (a joint $c$-coordinate extension whose embedding grows with $c$)
that we apply identically to every system. When a neighbour count vanishes the estimator is degenerate: if
$B_d(r)=0$ no base-embedding match exists and the ratio $B_{d+c}(r)/B_d(r)$ is undefined, and if
$B_d(r)>0$ but $B_{d+c}(r)=0$ the sample entropy diverges; in either case the affected trajectory is
excluded at that scale. Per-application parameters are in Table~\ref{tab:params}.

\begin{algorithm}[t]
\caption{Multivariate multiscale sample entropy $S^{(s)}$ of one trajectory.}\label{alg:mmse}
\begin{algorithmic}[1]
\Require channels $x^{(1)},\dots,x^{(c)}$ of length $N$; scale $s$; dimensions $\mathbf{m}=(m_1,\dots,m_c)$; per-channel delays $\boldsymbol{\tau}=(\tau_1,\dots,\tau_c)$; tolerance $r$
\For{$a = 1$ to $c$}
  \State $y^{(a)} \gets$ non-overlapping block means of $x^{(a)}$ over windows of length $s$ \Comment{coarse-grain}
  \State $y^{(a)} \gets (y^{(a)} - \mathrm{mean}(y^{(a)}))/\mathrm{std}(y^{(a)})$ \Comment{per-channel standardise}
\EndFor
\State build composite delay vectors $V_d$ of dimension $d=\sum_a m_a$ from $\{y^{(a)}\}$ with delay $\tau$
\State build $V_{d+c}$ by extending \emph{every} channel's embedding by one further delayed sample (composite dimension $d+c$)
\State $B_d(r) \gets$ fraction of distinct pairs of $V_d$ within Chebyshev distance $r$ \Comment{$k$-d tree, $O(N\log N)$}
\State $B_{d+c}(r) \gets$ fraction of distinct pairs of $V_{d+c}$ within $r$, over the same start indices
\State \Return $-\ln\!\big(B_{d+c}(r)/B_d(r)\big)$
\end{algorithmic}
\end{algorithm}

\begin{table}[t]
\centering\small
\caption{Per-application parameters: embedding $m$, delay $\tau$, tolerance $r$, window length $N$ (in
samples; physical units in parentheses), clustering threshold (relative fraction of the maximum linkage
distance, or absolute), and ensemble size $M$ (number of trajectories or segments). The chaos ensembles use the canonical chaotic
settings (Lorenz $\rho=28$; R\"ossler parameters $(0.2,0.2,5.7)$).}\label{tab:params}
\begin{tabular}{lcccccl}
\toprule
System & $m$ & $\tau$ & $r$ & window $N$ & threshold & $M$\\
\midrule
Chaos (Lorenz/R\"ossler) & $(2,2,2)$ & $(1,1,1)$ & 0.2  & 5000 & abs.\ 10000 & 100 per regime\\
Turbulence               & $(1,1,1)$ & $(1,1,1)$ & 0.15 & 200  & 0.3 & 100 per group ($\times250$)\\
Ocean drifters           & $(2,2)$   & $(1,1)$   & 0.15 & 120 (h) & 0.3 & 193, 252\\
Tropical cyclones        & $(2,2)$   & $(1,1)$   & 0.15 & 24 ($\times3$ h) & 0.3 & 757, 1193, 1694\\
Maritime AIS             & $(2,2)$   & $(1,1)$   & 0.15 & 20 ($\times1$ min) & 0.3 & 163, 436, 394\\
Orbits/debris            & $(2,2,2)$ & $(1,1,1)$ & 0.2  & 90 ($\times60$ s) & 0.3 & 200 per group\\
ECG cohorts              & $(2,2,2)$ & $(1,1,1)$ & 0.15 & 1000 (10 s) & 0.3 & 120 per cohort\\
Seismic (Tohoku)         & $(2,2,2)$ & $(1,1,1)$ & 0.15 & 1200 (60 s) & 0.3 & 18, 60, 108\\
Variable stars           & $(2)$     & $(1)$     & 0.2  & 1024 & 0.3 & 250 per class\\
Gait (neurodeg.)         & $(2,2)$   & $(1,1)$   & 0.15 & 1500 (5 s) & 0.3 & 156--240\\
GNSS (crustal motion)    & $(2,2,2)$ & $(1,1,1)$ & 0.15 & 500 (d) & 0.3 & 240\\
\bottomrule
\end{tabular}
\end{table}

\section{Proof of the consistency theorem}\label{s:proof}

We restate Theorem~2.1 of the main text. Fix a pattern resolution, so that a random trajectory takes
values in a finite set of patterns $i=1,\dots,k$ with fixed probabilities $p_i>0$ summing to one, and,
conditional on the pattern, let the within-trajectory complexities be independent and identically
distributed with finite mean $S_i=\mathbb{E}[\,\cdot\mid i\,]<\infty$. Draw $M$ trajectories independently. Let $\hat p_i$ be the
empirical frequency of pattern $i$ and $\hat S_i$ the sample mean of the complexities of the drawn members
of pattern $i$. Then
\[
\sum_{i=1}^{k}\big(-\hat p_i\ln\hat p_i\big)\hat S_i \;\xrightarrow{\text{a.s.}}\;
-\sum_{i=1}^{k} p_i\ln p_i\,S_i, \qquad M\to\infty .
\]

\emph{Proof.} The pattern counts $(M\hat p_1,\dots,M\hat p_k)$ are multinomial$(M;p_1,\dots,p_k)$, so by
the strong law of large numbers $\hat p_i\to p_i$ almost surely for each $i$. Because every $p_i>0$, the
membership count $M_i=M\hat p_i\to\infty$ almost surely, so the prefix averages of the i.i.d.\ within-pattern
complexities over this random but almost-surely divergent index converge, $\hat S_i\to S_i$ almost surely, by
the strong law of large numbers. These finitely many almost-sure statements hold simultaneously on a common
probability space, so $(\hat{\mathbf p},\hat{\mathbf S})\to(\mathbf p,\mathbf S)$ jointly almost surely. The
map $(\mathbf{p},\mathbf{S})\mapsto\sum_i(-p_i\ln p_i)S_i$ is continuous at the limit point (all $p_i>0$,
each $S_i$ finite), with $-p\ln p$ extended continuously to $p=0$ under the convention $0\ln 0=0$. By the
continuous mapping theorem the empirical functional converges almost surely to its population value. For $S_i\equiv1$ the limit is
$-\sum_i p_i\ln p_i$, the Shannon entropy, recovering the consistency of the empirical Shannon entropy.
\hfill$\square$

\noindent The theorem fixes the pattern resolution, treating the pattern assignment as a fixed measurable
map with fixed probabilities $p_i$; it therefore does not establish consistency of the sample-dependent
agglomerative (Ward) clustering used in practice, whose partition is estimated from the same sample. Under
a fixed cut height the number of patterns grows
with $M$ and no fixed $p_i>0$ exists, a regime we do not treat analytically and characterise only
empirically. The identical argument applies to the separable product $H(C)\bar S$, which is therefore
equally consistent, so the theorem does not distinguish the two functionals. Finally, each $\hat S_i$ is
itself a finite-length sample-entropy estimate, biased and variable at small $N$; the theorem concerns the
$M\to\infty$ limit at fixed estimator, and finite-$M$ behaviour is characterised empirically by the
subsampling intervals reported throughout.

\section{The functional family and the weighting}\label{s:family}

Write $h_i=-p_i\ln p_i$, $H(C)=\sum_i h_i$, $\bar S=\sum_i p_i S_i$, and let $\delta_{jj'}$ denote the
dissimilarity between trajectories $j$ and $j'$. Five functionals combine the two
ingredients in different ways (Table~\ref{tab:family}). CWMMSE is the weighted entropy $\sum_i h_i S_i$.

\begin{table}[t]
\centering\small
\caption{The functional family and what each captures.}\label{tab:family}
\begin{tabular}{lll}
\toprule
Functional & Formula & Captures\\
\midrule
mean MMSE $\bar S$        & $\sum_i p_i S_i$        & within-trajectory complexity only\\
cluster entropy $H(C)$    & $-\sum_i p_i\ln p_i$    & between-trajectory diversity only\\
Rao $Q$                   & $\tfrac{1}{M^2}\sum_{j,j'=1}^{M} \delta_{jj'}$ & pairwise dissimilarity diversity only\\
product $H(C)\,\bar S$    & $H(C)\,\bar S$          & both, separably\\
CWMMSE                    & $\sum_i (-p_i\ln p_i)\,S_i$ & both, jointly (weighted entropy)\\
\bottomrule
\end{tabular}
\end{table}

Three properties single out the weighted-entropy form. (i) Two degenerate limits: a single dominant
pattern ($p_1\to1$) sends $h_1\to0$ and CWMMSE$\to0$ however complex that pattern; uniformly trivial
patterns ($S_i\to0$) send CWMMSE$\to0$ however diverse the population. The mean violates the first and the
cluster entropy the second. (ii) Uniform-complexity reduction: if $S_i\equiv S$ then
CWMMSE$=S\,H(C)$, so for $k$ equal clusters CWMMSE$=S\ln k$. (iii) Coupling: in general
CWMMSE$-H(C)\bar S=\sum_i h_i(S_i-\bar S)=\operatorname{Cov}_p(S_i,-\ln p_i)$, the population covariance of complexity and surprisal that the separable product discards
(\S\ref{s:surr}).

A hand-built test (Table~\ref{tab:handbuilt}) makes the point: the mean cannot tell a redundant-complex
ensemble from a balanced one, the cluster entropy cannot tell a diverse-simple one from a balanced one,
and only the two-ingredient functionals (the separable product and CWMMSE, which coincide here because
within-cluster complexity is homogeneous) single out the ensemble that is high in both. The convex normalisation
$\mathrm{CWMMSE}/H(C)$ collapses back towards $\bar S$ and is rejected.

\begin{table}[t]
\centering\small
\caption{Hand-built ensembles of 50 trajectories. Only CWMMSE singles out the balanced ensemble that is
both diverse and individually complex.}\label{tab:handbuilt}
\begin{tabular}{lrrrr}
\toprule
Ensemble & $\bar S$ & $H(C)$ & $H(C)\bar S$ & CWMMSE\\
\midrule
diverse-simple (10 clusters, $S=0.5$)    & 0.50 & 2.30 & 1.15 & 1.15\\
redundant-complex (2 clusters, $S=2.0$)  & 2.00 & 0.69 & 1.39 & 1.39\\
balanced (10 clusters, $S=2.0$)          & 2.00 & 2.30 & 4.61 & \textbf{4.61}\\
\bottomrule
\end{tabular}
\end{table}

Rao's quadratic entropy~\cite{rao1982,ricottaszeidl2006,leinstercobbold2012} is the natural prior-art
competitor: a diversity weighted by pairwise dissimilarities. It is correct on neither of the paper's two controls, the
healthy-versus-infarction ECG contrast (\S\ref{s:ecg}) and the orbital null, an ensemble of geometrically
diverse but individually predictable orbits (\S\ref{s:robust}); a maximally spread but
dynamically trivial ensemble is maximally diverse to it. CWMMSE instead weights by intrinsic dynamical
complexity. Weighted entropy was introduced by Belis and Guia\c{s}u~\cite{belis1968,guiasu1971} and used
to group data into classes by Guia\c{s}u~\cite{guiasu1986}; the within-series weighted permutation
entropy~\cite{fadlallah2013} is its single-trajectory analogue.

\section{ECG benchmark in full}\label{s:ecg}

We draw $400$ size-50 subsamples without replacement from each of the 120-record NORM and MI cohorts
(PTB-XL, leads I, II, V2~\cite{wagner2020ptbxl,goldberger2000}), cluster at the stated resolution, and
record each functional; the area under the curve (AUC) is the probability the functional ranks a healthy subsample
above a diseased one (0.5 is chance, below 0.5 is the wrong direction). Table~\ref{tab:ecgsweep} gives the
sweep. Averaged MMSE misranks the cohorts at every resolution; Rao $Q$ misranks them more strongly;
cluster entropy, the product and CWMMSE separate them correctly for resolution $\ge0.3$. A paired bootstrap test of
the separation $|\mathrm{AUC}-0.5|$ at resolution 0.3, resampling the per-subsample AUC realisations and
Benjamini--Hochberg corrected across the four comparisons, finds CWMMSE significantly more separating than
the mean ($p<0.001$); it is statistically distinguishable from the product, but the difference is
negligible (a separation gain of $0.005$ in AUC, $p=0.015$ corrected). The case for CWMMSE over the
product therefore rests on the orbital control (\S\ref{s:robust}) and the coupling (\S\ref{s:surr}), not on
the cardiac AUC.

\begin{table}[t]
\centering\small
\caption{ECG cohort-separation AUC and effect size $|d|$ at three clustering resolutions, for five
functionals. Below 0.5 is the wrong direction.}\label{tab:ecgsweep}
\begin{tabular}{llrr}
\toprule
resolution & functional & AUC & $|d|$\\
\midrule
0.2 & mean MMSE & 0.24 & 1.00\\
0.2 & cluster entropy & 0.85 & 1.33\\
0.2 & Rao $Q$ & 0.00 & 3.37\\
0.2 & product & 0.42 & 0.27\\
0.2 & CWMMSE & 0.43 & 0.24\\
\midrule
0.3 & mean MMSE & 0.23 & 1.06\\
0.3 & cluster entropy & 0.995 & 2.88\\
0.3 & Rao $Q$ & 0.004 & 3.65\\
0.3 & product & 0.91 & 2.09\\
0.3 & CWMMSE & 0.92 & 2.14\\
\midrule
0.4 & mean MMSE & 0.21 & 1.14\\
0.4 & cluster entropy & 0.997 & 4.12\\
0.4 & Rao $Q$ & 0.01 & 3.45\\
0.4 & product & 0.99 & 3.46\\
0.4 & CWMMSE & 0.99 & 3.64\\
\bottomrule
\end{tabular}
\end{table}

Figure~\ref{fig:ecgrobust} shows this sweep visually, together with the orbital null control, an ensemble of
geometrically diverse but individually predictable orbits, that the
main-text controls scorecard summarises.

\begin{figure}[htbp]
\centering
\includegraphics[width=\textwidth]{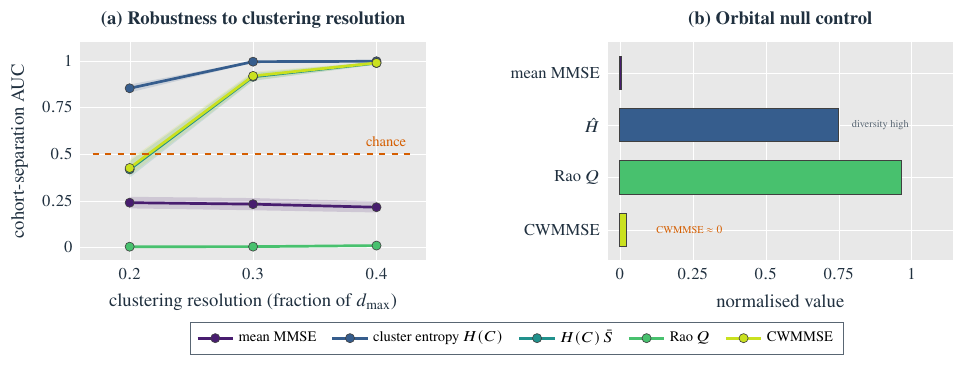}
\caption{Detail behind the controls scorecard of the main article. (a) Cohort-separation AUC (area under the ROC curve) against the clustering resolution (fraction of the maximum linkage distance) for the five functionals, with $95\%$ intervals and the chance value $0.5$ (dashed): the cluster entropy, the separable product and CWMMSE separate the cohorts on a broad plateau (resolution $\ge0.3$), whereas the mean MMSE and Rao $Q$ stay below chance throughout, so the cardiac ordering is not an artefact of one threshold. (b) The orbital null control on the breakup-debris ensemble, all values on a common $[0,1]$ scale: the diversity-only measures (normalised cluster entropy $\hat H$ and a normalised Rao $Q$) are large, whereas the mean MMSE and CWMMSE are near zero, because individual orbital motion is predictable.}\label{fig:ecgrobust}
\end{figure}

\section{Distance conventions and resolution robustness}\label{s:robust}

The clustering distance carries a scale and the cut height is a resolution; we verify the contrasts
against both. Table~\ref{tab:distance} and Figure~\ref{fig:esmdistance} compare three distance
conventions. The raw summed-Euclidean distance reproduces the main numbers; the scale convention
standardises each channel across the ensemble and averages over time, removing units and length while
preserving relative amplitude; the shape convention standardises each trajectory by its own statistics,
removing amplitude entirely. Every featured contrast holds under raw and scale (Table~\ref{tab:contrasts})
with sign probability $\ge0.99$. The shape convention over-resolves the clustering (almost
every short velocity window becomes its own pattern) and is reported as a documented sensitivity rather
than a recommended setting.

\begin{table}[t]
\centering\small
\caption{CWMMSE under the three distance conventions for featured groups. Raw and scale agree on every
ordering; shape over-resolves.}\label{tab:distance}
\begin{tabular}{lrrr}
\toprule
Group & raw & scale & shape\\
\midrule
Drifters Gulf Stream   & 2.08 & 2.19 & 4.43\\
Drifters Sargasso      & 2.86 & 2.86 & 3.96\\
ECG healthy            & 3.08 & 3.11 & 3.18\\
ECG MI                 & 1.51 & 2.32 & 3.43\\
Orbits Fengyun debris  & 0.02 & 0.02 & 0.02\\
Cyclones EP            & 2.09 & 1.95 & 4.06\\
Cyclones WP            & 0.79 & 0.80 & 1.44\\
GNSS boundary          & 0.42 & 0.97 & 1.74\\
Gait ALS               & 0.22 & 0.22 & 0.24\\
Stars eclipsing        & 0.08 & 0.08 & 0.08\\
\bottomrule
\end{tabular}
\end{table}

\begin{table}[t]
\centering\small
\caption{Contrast sign-probabilities from subsampling under the raw and scale conventions: the fraction
of subsamples in which the contrast holds in the stated direction, with the $95\%$ interval of the
difference in CWMMSE.}\label{tab:contrasts}
\begin{tabular}{lrlrl}
\toprule
Contrast & \multicolumn{2}{c}{raw} & \multicolumn{2}{c}{scale}\\
 & $p$ & diff interval & $p$ & diff interval\\
\midrule
Sargasso $>$ Gulf Stream      & 0.999 & [0.34, 1.09] & 0.996 & [0.19, 0.96]\\
ECG healthy $>$ MI            & 1.000 & [0.52, 1.75] & 0.989 & [0.15, 1.12]\\
Orbits debris $>$ fleet       & 1.000 & [0.02, 0.02] & 1.000 & [0.02, 0.02]\\
Seismic pre $>$ wave train    & 1.000 & [0.26, 0.29] & 1.000 & [0.26, 0.29]\\
GNSS boundary $>$ interior    & 1.000 & [0.27, 0.44] & 1.000 & [0.67, 1.03]\\
Stars eclipsing $>$ Cepheid   & 1.000 & [0.03, 0.04] & 1.000 & [0.03, 0.04]\\
Gait ALS $>$ control          & 1.000 & [0.02, 0.09] & 1.000 & [0.02, 0.09]\\
\bottomrule
\end{tabular}
\end{table}

\begin{figure}[htbp]
\centering
\includegraphics[width=\textwidth]{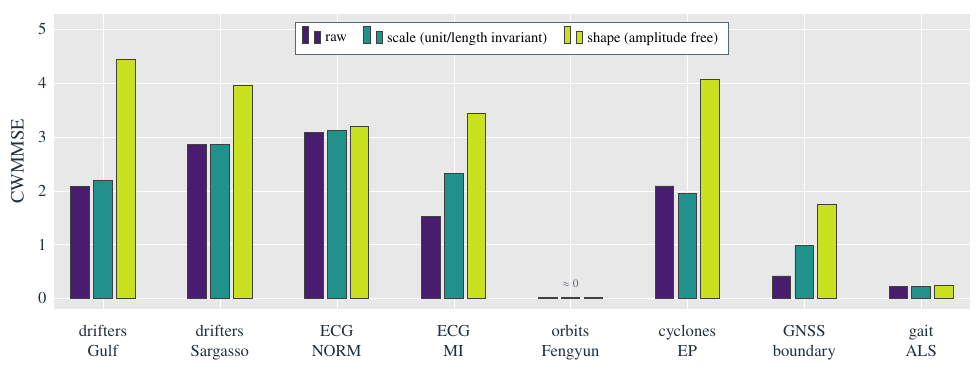}
\caption{CWMMSE is robust to the trajectory-distance convention used for clustering. Bars give CWMMSE for the eight featured groups (horizontal axis; the adjacent Table~\ref{tab:distance} lists two further groups), with three bars per group: \emph{raw} (summed point-wise Euclidean distance, the default used throughout), \emph{scale} (per-channel ensemble standardisation, removing units and overall amplitude), and \emph{shape} (per-trajectory standardisation, amplitude free). The raw and scale conventions give nearly identical CWMMSE and agree on every within-system ordering, so the reported contrasts are not an artefact of channel units or amplitude. The amplitude-free shape convention inflates CWMMSE for the short velocity-window ensembles (drifters, MI ECG, cyclones and GNSS), because discarding amplitude lets noise dominate the clustering and over-resolves the ensemble. It is included as a stress test and is not used for the results.}\label{fig:esmdistance}
\end{figure}

We check three further sources of variation. Sweeping the relative threshold over $[0.15,0.60]$ leaves
every contrast's direction unchanged on a broad plateau; the default 0.3 sits inside it. Parameter-free
cluster-number rules (silhouette, largest dendrogram gap) select two or three coarse clusters and flip the
cardiac contrast, because CWMMSE measures fine-grained diversity rather than the single dominant split.
The plug-in cluster entropy carries a downward finite-$M$ bias of order $(k-1)/(2M)$; we report the
Miller-Madow-corrected value $H_{\mathrm{mm}}=H(C)+(k-1)/(2M)$ alongside $H(C)$ in the released tables, and
it does not change any ordering.

The orbital null is the most discriminating control on diversity-only measures. The orbital ensembles reach the
highest normalised cluster entropy ($\hat H\approx0.75$) and the largest Rao $Q$ of any system, yet
CWMMSE$\approx0$ because individual orbital motion is predictable.

\paragraph{Linkage and entropy parameters.} The featured contrasts are insensitive to the agglomeration
rule~\cite{hennig2007}: under Ward, average and complete linkage the Sargasso, cardiac and orbital contrasts all hold in
the stated direction (Table~\ref{tab:sens}). They are likewise robust to the sample-entropy tolerance $r$
and embedding $m$ for the cardiac contrast, which holds for every combination of
$m\in\{(2,2,2),(3,3,3)\}$ and $r\in\{0.10,0.15,0.20\}$. The more marginal drifter contrast holds at the
default and larger tolerances (for example $r\ge0.15$ at $m=(2,2)$) but weakens at the smallest tolerance
and largest embedding, where the per-trajectory entropies, and hence the diversity structure, shift; we
report this sensitivity rather than tune it away.

\begin{table}[t]
\centering\small
\caption{Robustness of the featured contrasts (does CWMMSE rank the first group above the second?) to the
linkage rule, at the default resolution and raw distance.}\label{tab:sens}
\begin{tabular}{lccc}
\toprule
Linkage & Sargasso $>$ Gulf & healthy $>$ MI & debris $>$ fleet\\
\midrule
Ward     & yes & yes & yes\\
average  & yes & yes & yes\\
complete & yes & yes & yes\\
\bottomrule
\end{tabular}
\end{table}

\section{Surrogate and null tests}\label{s:surr}

Two nulls show CWMMSE responds to genuine structure. First, iterative amplitude-adjusted Fourier-transform
(IAAFT) surrogates~\cite{theiler1992,schreiber1996}, randomised control series matched to the data in
spectrum and amplitude but stripped of nonlinear structure, preserve each channel's power spectrum and
amplitude distribution while randomising phases. For a 60-trajectory subsample of the chaotic Lorenz
ensemble the real CWMMSE is 0.84 against a surrogate mean of 8.3 across $999$ surrogate ensembles, below
every one of them (one-sided $p=0.001$). The real value lies far below the surrogate band because deterministic motion is far more predictable per trajectory than
spectrum-matched noise; that the noise scores higher reflects the white-noise sensitivity of sample entropy
noted in the main text. CWMMSE is therefore not a function of the linear spectrum: the same separation is highly significant
for the chaotic R\"ossler and the electrocardiogram cohorts, which also fall below their surrogate bands
(each one-sided $p=0.001$; Figure~\ref{fig:esmsurr}a). The ocean drifters overlap their band and are not significant, consistent with their
velocity complexity being closer to linear. Second,
a complexity-permutation null holds the clustering fixed and permutes the per-trajectory complexities
across members; under the permutation each cluster's expected mean complexity is the grand mean $\bar S$,
so the null expectation of CWMMSE is $\sum_i h_i\bar S=H(C)\bar S$, the separable product, and a departure
of the real value from this null measures the coupling of equation~(2.5). Table~\ref{tab:coupling} shows the
coupling is significant in several systems: positive for the Sargasso ensemble, where the more complex
paths are also the more prevalent, and negative for the MI cohort, East-Pacific cyclones and the GNSS plate
boundary, where the most variable members are the least prevalent. It is negligible where cluster
complexities are homogeneous (deterministic chaos, the healthy ECG cohort, and the variable-star and gait
ensembles), for which CWMMSE and the separable product nearly coincide. This is the empirical case for the single weighted-entropy
functional over the separable product. Figure~\ref{fig:esmsurr} summarises both nulls.

\begin{table}[t]
\centering\small
\caption{Complexity-permutation null: real CWMMSE, the null mean (equal to the separable product), the
coupling (real minus null), its nominal two-sided $p$-value and the Benjamini--Hochberg-corrected value
$p_{\mathrm{BH}}$ (correction across all tested groups). Columns are rounded independently, so the displayed
coupling may differ from the difference of the rounded columns by one in the last digit. The GNSS coupling
does not survive correction.}\label{tab:coupling}
\begin{tabular}{lrrrrr}
\toprule
Group & CWMMSE & null mean & coupling & $p$ & $p_{\mathrm{BH}}$\\
\midrule
Drifters Sargasso     & 2.86 & 2.77 & $+0.10$ & 0.001 & 0.002\\
Drifters Gulf Stream  & 2.08 & 2.10 & $-0.02$ & 0.71 & 1.00\\
ECG healthy           & 3.08 & 3.08 & $-0.00$ & 0.90 & 1.00\\
ECG MI                & 1.51 & 1.71 & $-0.20$ & 0.01 & 0.04\\
Cyclones EP           & 2.09 & 2.23 & $-0.15$ & $<0.001$ & $<0.001$\\
GNSS plate boundary   & 0.42 & 0.50 & $-0.09$ & 0.04 & 0.09\\
Lorenz chaotic        & 0.95 & 0.95 & $+0.00$ & 1.00 & 1.00\\
\bottomrule
\end{tabular}
\end{table}

\begin{figure}[htbp]
\centering
\includegraphics[width=\textwidth]{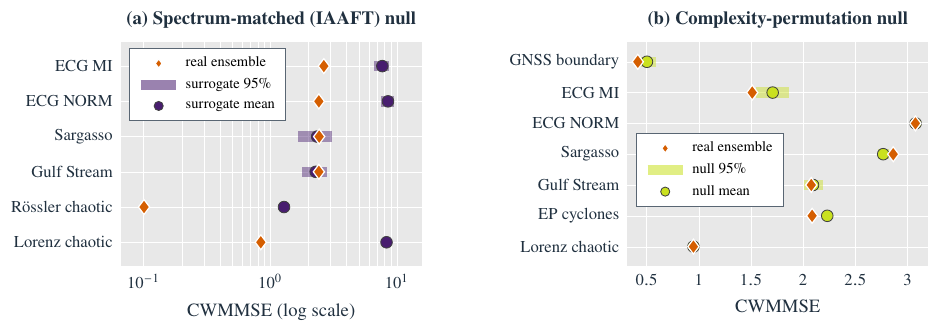}
\caption{Two null tests showing that CWMMSE responds to genuine structure, not to the linear spectrum or to chance. (a) Spectrum-matched (IAAFT) null: each trajectory channel is replaced by a surrogate with the same amplitude distribution and power spectrum but randomised phases, destroying nonlinear and deterministic structure. Orange diamonds mark the real-ensemble CWMMSE and the violet bars give the $2.5$--$97.5$ percentile surrogate band (note the logarithmic horizontal axis). The deterministic and physiological ensembles (chaotic Lorenz and R\"ossler, NORM and MI hearts) sit far below their surrogate bands, since real dynamics are far more predictable per trajectory than phase-randomised noise of the same spectrum; the ocean drifters overlap their band, consistent with their velocity complexity being close to linear and their contrast being carried by diversity rather than nonlinearity. (b) Complexity-permutation null specific to CWMMSE: the clustering is held fixed and the per-trajectory complexities are permuted across members, which preserves the cluster sizes and the mean complexity, so the null mean equals the separable product $H(C)\bar S$ (green bars, $95\%$ interval). A real value (orange diamond) lying outside this band signals a genuine coupling between which trajectories are complex and how the ensemble is partitioned, the term the separable product discards.}\label{fig:esmsurr}
\end{figure}

\cleardoublepage
\section{The eleven systems in detail}\label{s:systems}

Figure~\ref{fig:esmall} collects the CWMMSE of every group with its subsample interval, and
Table~\ref{tab:master} gives the full numbers. Each system is described below with its scientific background, data source,
preprocessing, CWMMSE result and interpretation; the illustrative figures
(Figures~\ref{fig:syschaos}--\ref{fig:sysgnss}) show what each system's raw data looks like. Throughout, we
read each contrast through the decomposition: a difference in CWMMSE may be driven by individual complexity
($\bar S$), by ensemble diversity ($\hat H$), or by both, and naming which is the point of the measure.

\begin{figure}[htbp]
\centering
\includegraphics[width=\textwidth]{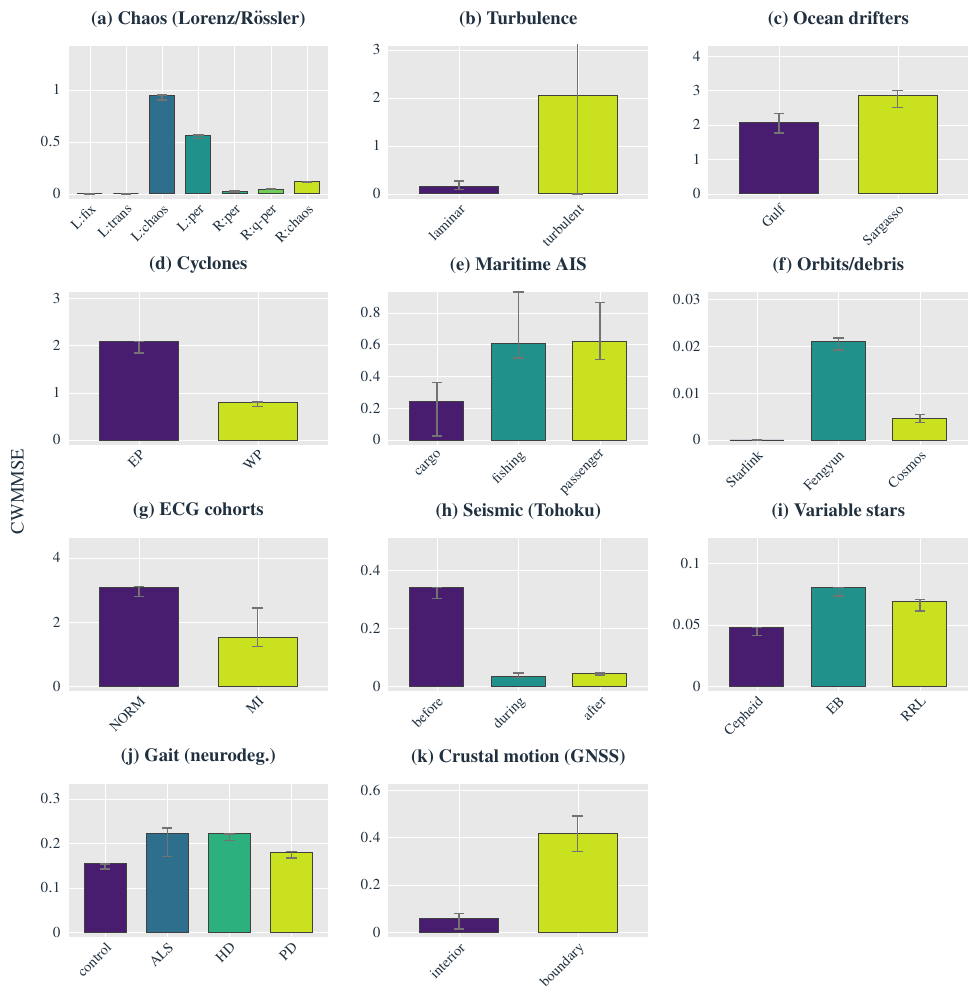}
\caption{CWMMSE for every group across all eleven systems, one panel per system (a--k), using the raw distance convention. Each panel shows the groups of one system (its dynamical regimes, geographic regions, vessel or orbital classes, clinical cohorts, basins or time windows); bar height encodes CWMMSE and the groups within a panel are drawn in distinct colours, while the whiskers are $2.5$--$97.5$ percentile subsample stability ranges, obtained by resampling $0.8M$ trajectories without replacement. The vertical scale differs between panels because CWMMSE is meaningful within a system rather than across systems; the corresponding numbers and intervals are listed in Table~\ref{tab:master}.}\label{fig:esmall}
\end{figure}

\begin{table}[t]
\centering\small
\caption{Full results (raw distance): CWMMSE with subsample interval, normalised cluster entropy
$\hat H$, mean MMSE $\bar S$, cluster count $k$, ensemble size $M$. $^{\dagger}$For the chaos ensembles
CWMMSE scales as $\bar S\ln k$ with $k\approx M$, so a size-$0.8M$ subsample yields a systematically smaller
value and the subsample-stability range lies below the full-sample point estimate rather than bracketing it;
these are stability ranges, not confidence intervals.}\label{tab:master}
\begin{tabular}{llrrrrr}
\toprule
System & Group & CWMMSE [stab.\ range] & $\hat H$ & $\bar S$ & $k$ & $M$\\
\midrule
Chaos & Lorenz chaotic & 0.95 [0.90,0.91]$^{\dagger}$ & 1.00 & 0.21 & 100 & 100\\
Chaos & Lorenz periodic & 0.57 [0.53,0.56]$^{\dagger}$ & 0.87 & 0.14 & 66 & 100\\
Chaos & R\"ossler chaotic & 0.12 [0.11,0.12] & 0.97 & 0.03 & 89 & 100\\
Turbulence & turbulent & 2.05 & 0.43 & 1.03 & n/a & 25000\\
Turbulence & laminar & 0.15 & 0.42 & 0.08 & n/a & 25000\\
Drifters & Sargasso & 2.86 [2.51,2.99] & 0.66 & 0.76 & 47 & 252\\
Drifters & Gulf Stream & 2.08 [1.76,2.33] & 0.46 & 0.87 & 22 & 193\\
Cyclones & East Pacific & 2.09 [1.84,2.03] & 0.53 & 0.60 & 68 & 1193\\
Cyclones & North Atlantic & 1.35 [1.21,1.39] & 0.54 & 0.41 & 51 & 757\\
Cyclones & West Pacific & 0.79 [0.71,0.81] & 0.57 & 0.20 & 101 & 1694\\
AIS & passenger & 0.62 [0.51,0.86] & 0.27 & 0.36 & 8 & 394\\
AIS & fishing/towing & 0.61 [0.52,0.93] & 0.27 & 0.38 & 12 & 436\\
AIS & cargo/tanker & 0.24 [0.03,0.36] & 0.15 & 0.33 & 6 & 163\\
Orbits & Fengyun-1C debris & 0.02 [0.02,0.02] & 0.75 & 0.005 & 60 & 200\\
Orbits & Cosmos-2251 debris & 0.005 [0.00,0.01] & 0.75 & 0.001 & 56 & 200\\
Orbits & Starlink & 0.00 & 0.74 & 0.00 & 54 & 200\\
ECG & healthy (NORM) & 3.08 [2.80,3.11] & 0.97 & 0.665 & 108 & 120\\
ECG & infarction (MI) & 1.51 [1.23,2.43] & 0.50 & 0.716 & 39 & 120\\
Seismic & before (noise) & 0.34 [0.30,0.32] & 1.00 & 0.118 & 18 & 18\\
Seismic & during (waves) & 0.03 [0.03,0.05] & 0.52 & 0.027 & 20 & 60\\
Seismic & after (coda) & 0.05 [0.04,0.05] & 0.71 & 0.015 & 51 & 108\\
Variable stars & Eclipsing binary & 0.08 [0.07,0.08] & 0.70 & 0.020 & 64 & 250\\
Variable stars & RR Lyrae & 0.07 [0.06,0.07] & 0.50 & 0.025 & 19 & 250\\
Variable stars & Cepheid & 0.05 [0.04,0.05] & 0.48 & 0.018 & 20 & 250\\
Gait & ALS & 0.22 [0.17,0.23] & 0.81 & 0.053 & 67 & 156\\
Gait & Huntington & 0.22 [0.21,0.22] & 0.83 & 0.048 & 113 & 240\\
Gait & Parkinson & 0.18 [0.17,0.18] & 0.80 & 0.043 & 76 & 180\\
Gait & control & 0.15 [0.14,0.15] & 0.72 & 0.041 & 52 & 192\\
GNSS & plate boundary & 0.42 [0.34,0.49] & 0.29 & 0.317 & 7 & 240\\
GNSS & plate interior & 0.06 [0.01,0.08] & 0.32 & 0.035 & 8 & 240\\
\bottomrule
\end{tabular}
\end{table}

\FloatBarrier  
\paragraph{Deterministic chaos (Lorenz, R\"ossler).} The Lorenz and R\"ossler systems are the textbook
route from order to chaos, so the complexity ordering across their regimes is known in advance and the
case is a validation rather than a discovery. We integrate ensembles of 100 trajectories per regime from
perturbed initial conditions to a long transient-free record~\cite{lorenz1963,rossler1976}. CWMMSE peaks
at the chaotic regime of each system (Lorenz $0.95$, R\"ossler $0.12$) and vanishes at the fixed point,
but the two peaks arise differently: the Lorenz peak is large in both ingredients ($\hat H\to1$,
$\bar S=0.21$), every trajectory individually irregular and no two alike, whereas the R\"ossler peak is
diversity-driven, its trajectories spreading into many distinct orbits ($\hat H=0.97$) while each stays
individually simple ($\bar S=0.03$). The main text develops this case; the
attractors are shown in Figure~\ref{fig:syschaos}.
\begin{figure}[H]\centering\includegraphics[width=\textwidth]{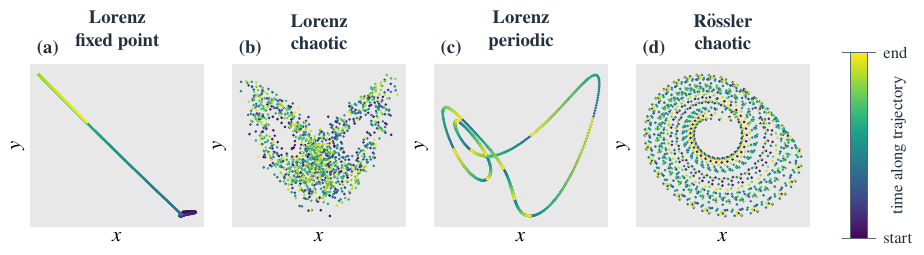}
\caption{Sample phase-space trajectories for the deterministic-chaos validation, one per regime, each a single ensemble member projected onto two state coordinates and coloured by time (dark early to bright late). (a) Lorenz fixed point: the path spirals inward and halts, with essentially no within-trajectory complexity. (b) Lorenz chaotic: the classic two-lobed attractor, an aperiodic orbit that never repeats. (c) Lorenz periodic: a closed limit cycle traced repeatedly. (d) R\"ossler chaotic: a stretched-and-folded spiral band. These single-member paths illustrate the regimes whose $100$-member ensembles are scored in Figure~\ref{fig:esmall}(a).}\label{fig:syschaos}\end{figure}

\paragraph{Fluid turbulence.} Tracking many passive particles through a flow turns a single velocity field
into an ensemble of Lagrangian trajectories, and the transition from laminar to turbulent motion is the
canonical increase in dynamical complexity: laminar streamlines are smooth and predictable, whereas
turbulence fills a broad band of interacting scales with intermittent, irregular fluctuations. We analyse
25000 Lagrangian particle paths from the CMRsim turbulent-flow simulation~\cite{weine2024cmrsim}, split
into 250 groups of 100, in a laminar and a turbulent regime. The two regimes
are almost equally diverse as populations ($\hat H=0.43$ turbulent against $0.42$ laminar), so what
separates them is individual complexity, more than an order of magnitude higher in the turbulent flow
($\bar S=1.03$ against $0.08$). CWMMSE inherits the gap, $2.05$ against $0.15$, and reads it as
complexity-driven: the turbulent ensemble is the more complex system because its members are individually
richer, not because they are more varied. Only the two steady regimes are available, so the case
illustrates the complexity-driven end of the spectrum (Figure~\ref{fig:systurb}).
\begin{figure}[H]\centering\includegraphics[width=\textwidth]{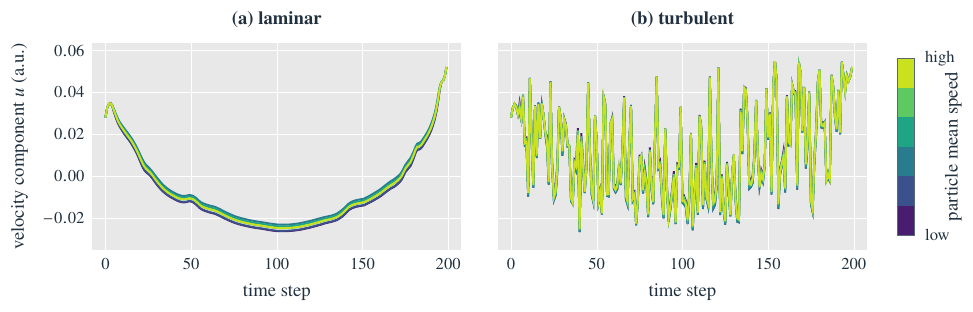}
\caption{Sample particle-velocity time series from the laminar and turbulent flow ensembles. Each panel shows six high-energy particles, plotting one in-plane velocity component against time step on a shared vertical scale. (a) In the laminar flow the velocities are smooth and nearly coincident between particles, giving low within-trajectory complexity. (b) In the turbulent flow each particle's velocity fluctuates rapidly and irregularly, the source of the much higher per-trajectory sample entropy that gives the turbulent ensemble its larger CWMMSE.}\label{fig:systurb}\end{figure}

\paragraph{Ocean drifters.} Surface drifters report their position hourly, so each is a Lagrangian probe
of the near-surface circulation~\cite{elipot2016drifter}; we take two regions in summer 2018 and analyse $120$-hour segments of each drifter's eastward and
northward velocity. This is the clearest diversity-driven reversal in the
corpus: the energetic Gulf Stream is individually the more complex ($\bar S=0.87$ against $0.76$), yet the
quiescent Sargasso is the more complex \emph{system} ($\hat H=0.66$ against $0.46$), because the boundary
current entrains its drifters into a few shared structures while the gyre interior supports many
independent paths. CWMMSE follows the diversity and ranks Sargasso above Gulf Stream ($2.86$ against
$2.08$); the main text develops the contrast (Figure~\ref{fig:sysdrift}).
\begin{figure}[H]\centering\includegraphics[width=82mm]{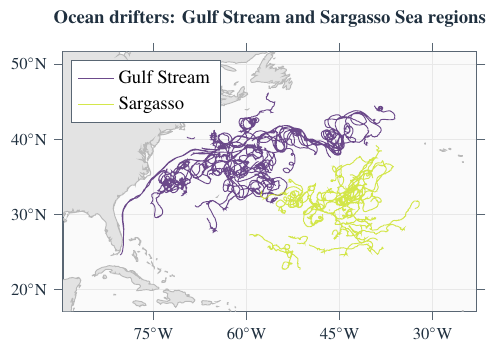}
\caption{Full trajectories of the NOAA Global Drifter Program surface drifters used for the ocean contrast, on a North Atlantic map (land grey, coastlines drawn, longitude and latitude on the axes). Gulf Stream drifters (violet) are entrained by the western-boundary current and swept rapidly northeastward along a few shared paths, whereas Sargasso drifters (yellow-green) meander slowly and individually through the gyre interior. The analysis uses $120$-hour windows of the two velocity components while each drifter is within its region. Because the Gulf Stream paths are shared and coherent, that ensemble is the less diverse one and hence the lower in CWMMSE, despite its individually more tortuous tracks.}\label{fig:sysdrift}\end{figure}

\paragraph{Tropical cyclones.} A best-track is the time series of a storm centre, and ocean basins differ
systematically in how their storms move, from long recurving paths in some basins to more zonal tracks in
others. We take IBTrACS best-track velocities~\cite{knapp2010} for storms since 2000 in three basins, in
$24\times3$-hour windows. The three basins span almost the same ensemble diversity ($\hat H$ between $0.53$ and
$0.57$), so the ordering is set by individual track irregularity: mean complexity falls from the East
Pacific ($\bar S=0.60$) through the North Atlantic ($0.41$) to the West Pacific ($0.20$). CWMMSE follows in
the same order ($2.09$, $1.35$, $0.79$), a complexity-driven contrast like turbulence and crustal motion.
The tracks are smooth and short relative to the other systems, so the case is exploratory rather than
definitive (Figure~\ref{fig:syscyc}).
\begin{figure}[H]\centering\includegraphics[width=\textwidth]{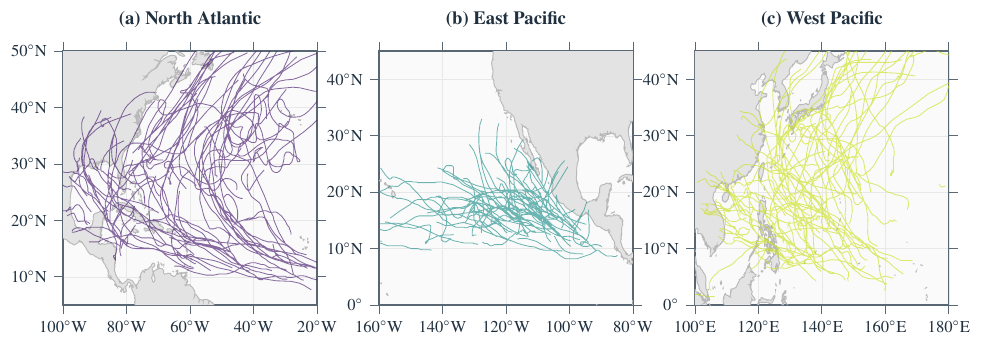}
\caption{Tropical-cyclone best-track trajectories from IBTrACS, one square panel per ocean basin: North Atlantic, East Pacific and West Pacific. Each line is one storm's centre track, with coastlines drawn and longitude and latitude on the axes. The basins differ in how varied and irregular their storm paths are, the property CWMMSE quantifies: the East Pacific tracks are individually the most erratic and take the highest CWMMSE, ahead of the North Atlantic and the West Pacific. Per-basin values and intervals are in Table~\ref{tab:master}.}\label{fig:syscyc}\end{figure}

\paragraph{Maritime vessels.} Ships broadcast their position over AIS, and vessel classes move very
differently: cargo ships and tankers follow regulated, repeatable shipping lanes, whereas fishing, towing
and passenger vessels manoeuvre and loiter. We take one day of AIS traffic in the Los Angeles and Long
Beach approaches~\cite{marinecadastre}, in $20\times1$-minute windows by vessel category. The cargo and
tanker class is the least complex on both ingredients, and most of the gap is diversity: it occupies far
fewer behavioural patterns ($\hat H=0.15$) than the manoeuvring classes ($\hat H\approx0.27$), at similar
individual complexity ($\bar S$ between $0.33$ and $0.38$ across classes). CWMMSE accordingly ranks the
manoeuvring vessels above the lane-followers ($0.62$ for passenger and $0.61$ for fishing or towing against
$0.24$ for cargo and tankers); the corresponding main-text figure groups the two manoeuvring classes
against the lane-following one (Figure~\ref{fig:sysais}).
\begin{figure}[H]\centering\includegraphics[width=82mm]{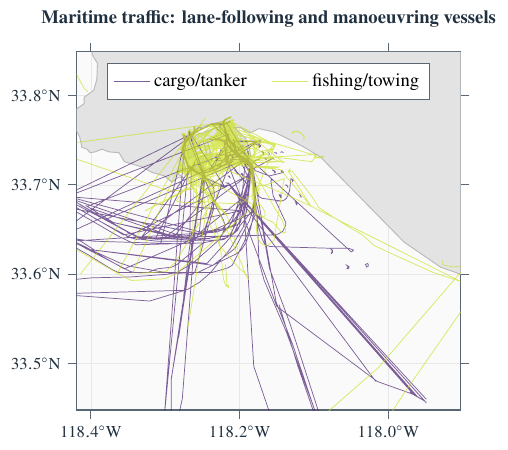}
\caption{Full vessel tracks near the ports of Los Angeles and Long Beach, coloured by vessel class, from Automatic Identification System (AIS) data (Marine Cadastre): cargo and tanker traffic (violet), which follows straight, repeatable shipping lanes, against fishing and towing vessels (yellow-green), which manoeuvre and loiter in irregular patterns (land in grey, coastlines drawn, longitude and latitude on the axes). The manoeuvring classes occupy a more varied set of behavioural patterns than the lane-followers, at comparable individual complexity, which CWMMSE registers as the higher score (Figure~\ref{fig:esmall}(e)).}\label{fig:sysais}\end{figure}

\paragraph{Orbits and debris.} A satellite catalogue propagated forward in time is an ensemble of orbital
trajectories~\cite{celestrak,vallado2006}; we contrast an intact Starlink constellation with two breakup
debris clouds, in $90\times60$-second windows. These ensembles carry the highest diversity of any system
($\hat H\approx0.75$, the objects spread across many distinct orbital planes) yet the lowest individual
complexity ($\bar S\lesssim0.005$), so CWMMSE is near zero throughout. This is the null control of the main
text: the one case a diversity-only measure would wrongly rank as the most complex system of all
(Figure~\ref{fig:sysorb}).
\begin{figure}[H]\centering\includegraphics[width=\textwidth]{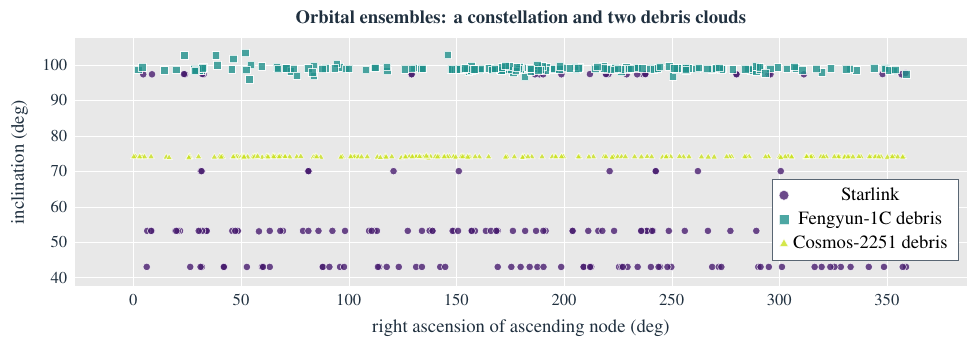}
\caption{Orbital geometry of the tracked objects: the orbital null. Each point is one object from CelesTrak two-line elements propagated with SGP4, with the right ascension of the ascending node (the orbit-plane orientation) on the horizontal axis and the inclination on the vertical axis. The Starlink constellation (violet) occupies a single tightly controlled inclination band near $53^{\circ}$, whereas the Fengyun-1C (teal, near $99^{\circ}$) and Cosmos-2251 (yellow, near $74^{\circ}$) debris clouds spread across all ascending-node orientations. The debris is geometrically diverse (high ensemble diversity $\hat H$) yet each fragment follows a simple near-Keplerian orbit (low individual complexity $\bar S$), so its CWMMSE is near zero. This is the null control of the main article.}\label{fig:sysorb}\end{figure}

\paragraph{Electrocardiogram cohorts.} Each 10-second 12-lead record (leads I, II and V2) from the PTB-XL
database~\cite{wagner2020ptbxl,goldberger2000} is one multivariate trajectory, and the ensemble is a
clinical cohort. The healthy (NORM) and myocardial-infarction (MI) cohorts have nearly equal individual
complexity, with the diseased cohort marginally higher ($\bar S=0.716$ against $0.665$), so the mean MMSE
ranks the cohorts in the wrong direction; the separation lives entirely in diversity, which contracts with disease
($\hat H=0.97$ healthy against $0.50$ infarction). CWMMSE reads this correctly, separating the cohorts by
roughly a factor of two ($3.08$ against $1.51$), the principal result the main text develops
(Figure~\ref{fig:sysecg}).
\begin{figure}[H]\centering\includegraphics[width=\textwidth]{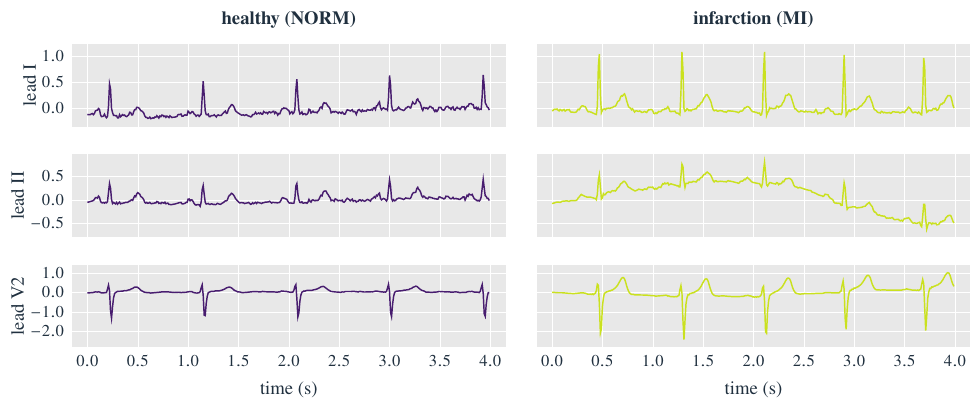}
\caption{Sample $4$-second electrocardiograms from the PTB-XL database. Each panel shows leads I, II and V2 (rows) for a healthy record (NORM, violet, left column) and a myocardial-infarction record (MI, yellow-green, right column). Both records show the regular QRS spikes of the heartbeat; the NORM and MI cohorts differ in subtler waveform morphology and baseline behaviour rather than in obvious irregularity, which is why their mean within-trajectory complexity is nearly equal and they are separated instead by ensemble diversity. The cohort CWMMSE values are in Figure~\ref{fig:esmall}(g).}\label{fig:sysecg}\end{figure}

\paragraph{Seismic record.} A great earthquake briefly imposes coherent order on the seismic wavefield:
before the event a station records ambient noise from many uncorrelated sources, whereas the main shock
arrives as a single coherent wave train shared by every window of the record. We window a three-component
IU.ANMO recording of the 2011 Tohoku earthquake into $60$-second segments~\cite{beyreuther2010} and form
ensembles from the pre-event, main-shock and coda phases. The pre-event noise is the more complex system
on both ingredients, being both more diverse ($\hat H=1.00$ against $0.52$ during the wave train) and
individually more irregular ($\bar S=0.12$ against $0.03$). CWMMSE therefore falls sharply across the
event, from $0.34$ before to $0.03$ during, registering the earthquake as an ordering event that
simplifies the ensemble rather than as a burst of complexity (Figure~\ref{fig:sysseis}).
\begin{figure}[H]\centering\includegraphics[width=\textwidth]{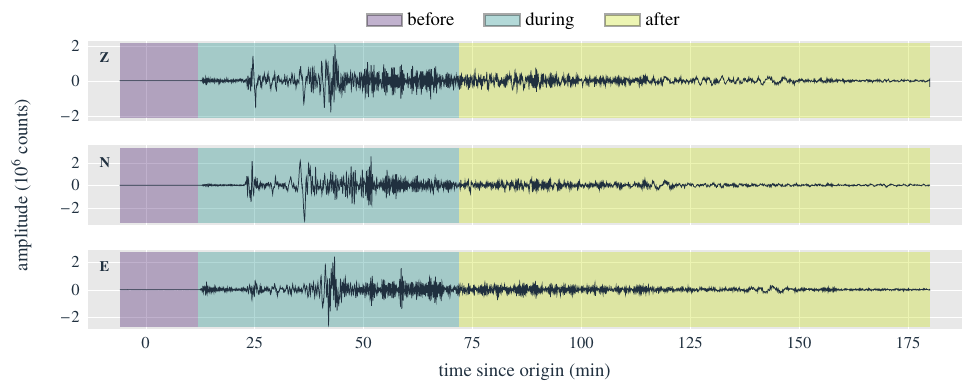}
\caption{Three-component broadband seismogram of the 2011 Tohoku earthquake, recorded at station IU.ANMO (vertical Z, north N and east E components), shown as amplitude against time since origin. Three analysis windows are shaded: \emph{before} (the quiet pre-event noise), \emph{during} (the high-amplitude body- and surface-wave arrivals) and \emph{after} (the decaying coda). Ensembles are formed from many such windowed segments. The pre-event noise is the most diverse ensemble of segments, whereas the strong main-shock wave train is dominated by a few coherent arrivals, so CWMMSE is highest \emph{before} the event (Figure~\ref{fig:esmall}(h)).}\label{fig:sysseis}\end{figure}

\paragraph{Variable-star light curves.} A variable star is catalogued by its light curve, the periodic
rise and fall of its brightness, and distinct variability classes leave distinct signatures: the geometric eclipses of eclipsing
binaries and the radial pulsations of RR~Lyrae and Cepheids form three classical, expert-separable classes. We take the UCR
StarLightCurves set~\cite{dau2019ucr}, phase-folded normalised brightness series of length $1024$ with
$250$ examples per class. Although each class is individually a low-complexity periodic signal, the ranking by individual complexity and the
ranking by system complexity diverge: the RR~Lyrae are
individually the most variable ($\bar S=0.025$ against $0.020$ for eclipsing binaries), yet the
eclipsing-binary population is the more diverse ($\hat H=0.70$ against $0.50$) and so the more complex
\emph{system}. CWMMSE follows the diversity, ranking eclipsing binaries above RR~Lyrae and the smoother
Cepheids ($0.08$, $0.07$, $0.05$), a diversity-driven reversal that mirrors the ocean drifters
(Figure~\ref{fig:sysstar}).
\begin{figure}[H]\centering\includegraphics[width=\textwidth]{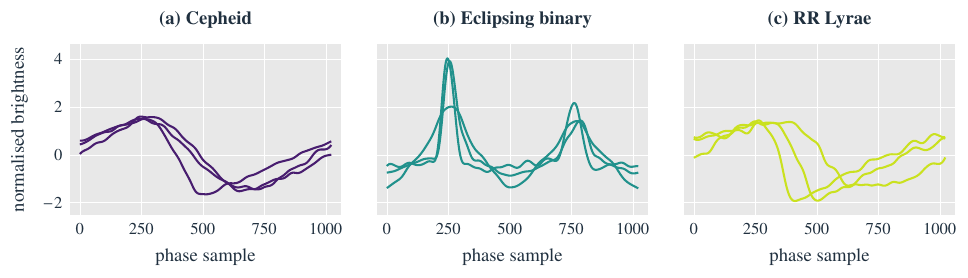}
\caption{Sample variable-star light curves, three per expert class (UCR StarLightCurves): (a) Cepheid, (b) eclipsing binary, (c) RR~Lyrae, each a normalised brightness series against phase sample (shared vertical scale). The light curves are individually smooth and near-periodic, so within-trajectory complexity is low and similar across classes; the classes differ mainly in how diverse the population of shapes is, which is the diversity ingredient that CWMMSE reads as ensemble complexity (Table~\ref{tab:master}).}\label{fig:sysstar}\end{figure}

\paragraph{Neurodegenerative gait.} Walking is a rhythmic motor behaviour, and neurodegenerative disease
disrupts its regularity, so the stride-to-stride dynamics of the ground-reaction force carry a clinical
signal. We take the PhysioNet gaitndd database~\cite{hausdorff2000,goldberger2000}, the vertical
ground-reaction force under each foot at $300$~Hz, normalised per record and cut into $5$-second windows,
with cohorts for Parkinson's disease, Huntington's disease, amyotrophic lateral sclerosis and healthy
controls. CWMMSE rises monotonically from controls ($0.15$) through Parkinson's ($0.18$) to Huntington's
and amyotrophic lateral sclerosis (both $0.22$), and here both ingredients move together: the disease
cohorts are individually more irregular ($\bar S$ rising from $0.041$ to $0.053$) and more diverse
($\hat H$ from $0.72$ to above $0.80$). The case is a clinical sibling of the electrocardiogram cohorts
with the sign reversed: where infarction contracts a cohort's complexity, neurodegeneration enlarges it,
a direction the decomposition reports rather than assumes
(Figure~\ref{fig:sysgait}).
\begin{figure}[H]\centering\includegraphics[width=\textwidth]{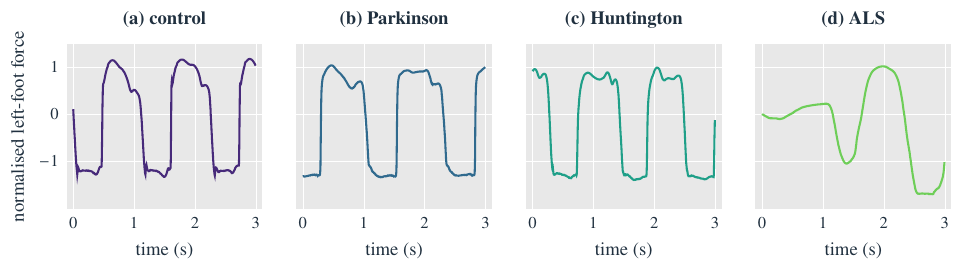}
\caption{Sample left-foot ground-reaction-force traces (normalised), one recording per group (PhysioNet gaitndd), showing a three-second excerpt of the five-second analysis windows. (a) control, (b) Parkinson's, (c) Huntington's, (d) amyotrophic lateral sclerosis. The neurodegenerative cohorts (b) to (d) show more irregular stride-to-stride force than the control (a), the source of their higher per-trajectory complexity and larger CWMMSE (Table~\ref{tab:master}).}\label{fig:sysgait}\end{figure}

\paragraph{Crustal motion (GNSS).} A continuously operating GNSS receiver fixes its own position daily, so
a multi-decade station record traces how the solid Earth moves beneath it. In a stable plate interior that
motion is little more than seasonal loading and monument noise, whereas near an active plate boundary it
carries secular strain, transient slip and the response to nearby earthquakes. We take daily East, North
and Up series from the Nevada Geodetic Laboratory~\cite{blewitt2018ngl} in the North-America plate-fixed
frame, contrasting interior stations with stations along the active western boundary, in $500$-day windows.
The two populations are almost equally diverse ($\hat H=0.29$ at the boundary against $0.32$ in the
interior), so what separates them is individual complexity, an order of magnitude higher at the boundary
($\bar S=0.32$ against $0.035$). CWMMSE separates boundary from interior by about a factor of seven ($0.42$
against $0.06$) and reads the gap as complexity-driven rather than diversity-driven, the mirror image of
the drifter and variable-star contrasts (Figure~\ref{fig:sysgnss}).
\begin{figure}[H]\centering\includegraphics[width=\textwidth]{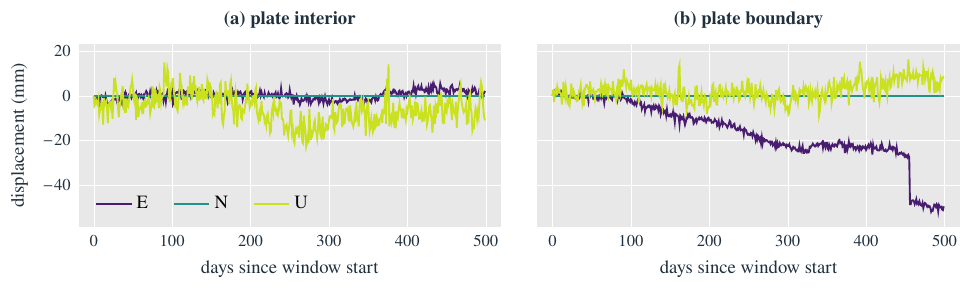}
\caption{Sample GNSS (Global Navigation Satellite System) daily-position series for two stations. Each panel shows East, North and Up displacement from the window start (in mm) for (a) a plate-interior station and (b) a plate-boundary station, in the North-America plate-fixed frame (Nevada Geodetic Laboratory). The interior station (a) varies little and smoothly, whereas the boundary station (b) carries structured tectonic motion; this is the source of its far higher per-trajectory complexity and CWMMSE (Table~\ref{tab:master}).}\label{fig:sysgnss}\end{figure}

\FloatBarrier  
\section{Data and code availability}\label{s:data}

All code is an openly available Python package, with scripts that regenerate the figures and reported numbers
(\url{https://github.com/cisgroup/cwmmse}; DOI and Zenodo will be added on acceptance). Table~\ref{tab:datasets} lists the
datasets; all are open.

\begin{table}[t]
\centering\small
\caption{Datasets, sources and access.}\label{tab:datasets}
\begin{tabular}{lll}
\toprule
System & Source & Access\\
\midrule
Drifters & NOAA Global Drifter Program~\cite{elipot2016drifter} & ERDDAP, open\\
Cyclones & NOAA IBTrACS v04~\cite{knapp2010} & download, open\\
AIS & Marine Cadastre~\cite{marinecadastre} & download, open\\
Orbits & CelesTrak TLE~\cite{celestrak} & download, open\\
ECG & PhysioNet PTB-XL~\cite{wagner2020ptbxl,goldberger2000} & PhysioNet, open\\
Seismic & IRIS/EarthScope FDSN~\cite{beyreuther2010} & FDSN web service, open\\
Variable stars & UCR StarLightCurves~\cite{dau2019ucr} & download, open\\
Gait & PhysioNet gaitndd~\cite{hausdorff2000,goldberger2000} & PhysioNet, open\\
GNSS & Nevada Geodetic Lab~\cite{blewitt2018ngl} & download, open\\
\bottomrule
\end{tabular}
\end{table}

\printbibliography[title={Supplementary References}]
\end{refsection}

\end{document}